\documentclass[a4paper,UKenglish,cleveref,autoref]{lipics-v2021}
\usepackage{amsmath,amssymb,mathtools,xcolor}
\usepackage{thmtools,thm-restate}
\usepackage{microtype}
\usepackage[T1]{fontenc}
\usepackage{lmodern}
\hideLIPIcs
\title{Finite-on-Top-of-Temporal  CSPs \\ in Fixed-Point Logic}
\titlerunning{Finite on Top of Temporal CSPs in Fixed-Point Logic}
\author{Manuel Bodirsky}{Institut f\"ur Algebra, TU Dresden, Germany}{manuel.bodirsky@tu-dresden.de}{https://orcid.org/0000-0001-8228-3611}{}
\author{Philipp A. Grzywaczyk}{Institut f\"ur Algebra, TU Dresden, Germany}{philipp_alexander.grzywaczyk@tu-dresden.de}{https://orcid.org/0009-0007-5316-9386}{}
\authorrunning{M. Bodirsky and P. A. Grzywaczyk}
\Copyright{Manuel Bodirsky and Philipp A. Grzywaczyk}
\funding{Both authors have been funded by the European Research Council (Project POCOCOP, ERC Synergy Grant 101071674). Views and opinions expressed are however those of the authors only and do not necessarily reflect those of the European Union or the European Research Council Executive Agency. Neither the European Union nor the granting authority can be held responsible for them.}
\acknowledgements{The authors would like to thank Antoine Mottet for the discussions on the topic of this paper.}
\makeatletter
\renewcommand{\p@enumii}{}
\makeatother
\EventEditors{}\EventNoEds{0}\EventLongTitle{}
\EventShortTitle{}\EventAcronym{}\EventYear{}
\EventDate{}\EventLocation{}
\SeriesVolume{}\ArticleNo{}
\newcommand{\conditionref}[1]{\textup{(\ref{#1})}}

\newcommand{\ignore}[1]{}
\newcommand{\FP}{\mathrm{FP}}
\newcommand{\FPC}{\mathrm{FPC}}
\newcommand{\FPnum}{\mathrm{FP}^{*}}
\newcommand{\B}{\mathfrak B}
\newcommand{\D}{\mathfrak D}
\newcommand{\E}{\mathfrak E}
\newcommand{\Pf}{\mathfrak P}
\newcommand{\Q}{\mathbb Q}
\newcommand{\bS}{\mathfrak S}

\newcommand{\A}{\mathfrak A}
\newcommand{\bD}{\mathfrak D}

\newcommand{\mi}{\mathrm{mi}}
\newcommand{\tll}{\mathrm{ll}}

\newcommand{\ppop}{\mathrm{pp}}
\DeclareMathOperator{\Pol}{Pol}
\DeclareMathOperator{\Inv}{Inv}
\DeclareMathOperator{\Aut}{Aut}
\DeclareMathOperator{\CSP}{CSP}
\DeclareMathOperator{\Sg}{Sg}
\DeclareMathOperator{\ar}{ar}
\DeclareMathOperator{\pr}{pr}

\nolinenumbers
\begin{document}
\maketitle
\begin{abstract}
We study definability in fixed-point logic with counting (FPC) for a class of constraint
satisfaction problems (CSPs) that simultaneously generalises  finite-domain CSPs and temporal CSPs. 
The templates for the CSPs 
are first-order expansions of 
the structure $([k] \times {\mathbb Q};<_1,\dots,<_k,\sim)$ where
$<_i \, \coloneqq \{((i,x),(i,x')) \mid x<x'\}$ and $\sim \, \coloneqq \{((i,x),(i,x')) \mid i \in [k], x,x' \in {\mathbb Q}\}$, i.e., $k$ copies of $({\mathbb Q};<)$. 

Such a CSP is definable in FPC precisely when the
template pp-constructs neither the ternary temporal relation $X$, which holds on three elements if the minimum
is attained exactly twice, nor a nontrivial finite affine structure.
Equivalently, the template has global polymorphisms that  witness the tractability of the temporal CSPs on the copies, 
and that provide 3-4 weak near-unanimity operations on the set of copies of ${\mathbb Q}$. 
In this case, we prove that
the CSP is even expressible in the less expressive logic FP$^*$, which is the restriction of FPC without counting terms (i.e., the extension of first-order logic 
by a fixed
point operator and an ordered auxiliary number sort). 

Our algorithm samples finitely many values in each copy and then performs nested
constraint-singleton consistency tests (which is a uniform polynomial-time algorithm for finite-domain CSPs). The correctness proof uses finite-domain CSP theory to combine consistency certificates by a single
polymorphism term, fixing each variable to one copy without destroying
the certificates needed for the temporal algorithm. 
We then use a known normal form for algebraic products to control the remaining interactions between different copies.
\end{abstract}

\medskip 
{\bf{AI Disclosure.}}
OpenAI's ChatGPT was used for initial algorithm and proof development, as well as drafting an initial LaTeX file. It contributed the ideas of the outer colour-pair testing in the algorithm, and using the syntactical normal forms for algebraic products of $(\Q;<)$. 
Many of the other ideas, e.g., using the combination of sampling and consistency algorithms, determining the colours via cycle-consistency of the induced finite CSP, and using palette polymorphisms came from the authors, and some others emerged during discussions with ChatGPT.
The authors reviewed and revised all arguments produced by LLMs; the proofs have been streamlined and rewritten. In particular, the entire text of the paper has been written by the authors. The authors verified the mathematical content, claims and citations and take full responsibility for the paper.

%%%%%%%%%%%%%%%%%%%%%%%%%%%%%%%%%%%%%%%%%%%%%%%%%%%%%%%%%%%%%%%%%%%%%%%%%%%%%%%%%%%%%%%%%%%%%%%%%%%%%%%%%%%%%%%%%%%%%%%%%%%%%%
\begingroup
\setcounter{tocdepth}{1}
\newpage
\tableofcontents
\endgroup

%\clearpage
%\setcounter{page}{1}
\section{Introduction}\label{sec:introduction}
A constraint satisfaction problem (CSP) is a computational problem
where we are given a finite set of variables and a finite set of constraints, and the question 
is 
whether variables can be assigned
values in a fixed relational structure so that all the  constraints are satisfied. 
One of the central research topics about CSPs is their \emph{descriptive complexity}  depending on the type of constraints that is allowed in the input (see, e.g., the survey article~\cite{Dualities}).
A prominent formalism in descriptive complexity is \emph{fixed-point logic} (FP),
a powerful extension of first-order logic~\cite{Libkin}. Every problem that can be expressed in fixed-point logic is \emph{in P}, i.e., can be solved in polynomial time. 
And indeed, many of the important polynomial-time algorithms for CSPs that are based on establishing some form of local consistency can be formulated in fixed-point logic. 
However, it is known that some CSPs in P cannot be expressed in fixed-point logic: this is for instance the case for the problem of deciding whether a given system of linear equations has a solution over the two-element field~\cite{AtseriasBulatovDawar} (which can be solved in polynomial time by Gaussian elimination). 

To discuss the current state of the art, it will be convenient to formulate CSPs as \emph{homomorphism problems}: for a fixed relational structure $\B$, 
called the \emph{template}, 
we view $\CSP(\B)$ as the class of all finite structures $\A$ that have a homomorphism to $\B$. We may now ask:  
for which structures $\B$ is there a sentence $\Phi$ in FP such that a finite structure $\A$ has a homomorphism to $\B$ if and only if $\A$ satisfies $\Phi$? If this is the case, we say that $\CSP(\B)$ is \emph{in FP} (and use analogous conventions for other logics discussed below). 

If $\B$ is a finite structure, it turns out that $\CSP(\B)$ is in FP if and only if it is in Datalog (an extremely restricted \emph{fragment} of FP) if and only if it is in FPC (an extremely powerful \emph{extension} of FP by a counting mechanism)~\cite{AtseriasBulatovDawar,BartoKozik}. 
In general, these logics are known to have different expressive power: for instance, \emph{Parity}, which is the class of all structures whose domain has even cardinality, is expressible in FPC, but not in FP. 
This stunning gap in the descriptive complexity of finite-domain CSPs coincides with a fundamental border in universal algebra, which can be characterized in numerous different ways. 
To state one of these characterisations, which will be of fundamental importance for our results, we need the concept of a  \emph{polymorphism} of $\B$; this is an operation
$B^s\to B$ preserving every relation from the signature of $\B$ coordinatewise.
Focussing on formulations that will be important later, we mention that the following are equivalent:
\begin{itemize}
    \item $\CSP(\B)$ is in FP; 
    \item $\CSP(\B)$ is in Datalog; 
    \item for every prime $p$, the structure $\B$ cannot pp-construct $({\mathbb Z}_p;+,1)$; 
    \item $\B$ has polymorphisms 
    %(see  Section~\ref{sect:pol}) 
    $g_1 \colon B^3 \to B$ and $g_2 \colon B^4 \to B$ that 
    are \emph{3-4 weak near-unanimity operations}, i.e., they satisfy, 
    for all $x,y \in B$,
    \begin{align*}
        & g_1(x,x,y) = g_1(x,y,x) = g_1(y,x,x) \\ = \; & g_2(x,x,x,y) = g_2(x,x,y,x) = g_2(x,y,x,x) = g_2(y,x,x,x).
    \end{align*}
\end{itemize}

If $\B$ is an arbitrary structure over an infinite domain, all of the mentioned equivalences fail~\cite{RydvalDescr}. They even fail if we focus on a model-theoretically very tame class of templates, namely on \emph{temporal structures}, which are structures with domain ${\mathbb Q}$ whose relations are (parameter-free) first-order definable in $({\mathbb Q};<)$. 
We give two concrete examples of temporal structures, because their CSPs, which are called \emph{temporal CSPs}, exhibit different descriptive complexities.

\begin{example}
    Let $R_{\min} \coloneqq \{(x,y,z) \in {\mathbb Q}^3 \mid x > y \vee x > z \}$. 
    Then $\CSP({\mathbb Q};R_{\min})$ is in FP~\cite{RydvalDescr}, but not in Datalog~\cite{Book}. 
\end{example}

\begin{example}
    Let $X = \{(x,y,z) \in {\mathbb Q}^3 \mid (x=y \wedge y< z) \vee (y=z \wedge z <x) \vee (z=x \wedge x<y) \}$. 
    Then $\CSP({\mathbb Q};X)$ is in P~\cite{tcsps-journal}, but not in FPC~\cite{RydvalDescr}. 
\end{example}

Also none of the algebraic characterisations mentioned above capture containment in Datalog, FP, or FPC (and are also inequivalent).
Nonetheless, the descriptive complexity of temporal CSPs has been classified.
Also for temporal CSPs, there is a gap in the descriptive complexity landscape: they are either in FP, or not even in FPC~\cite{RydvalDescr}. 
For arbitrary CSPs, it is an open problem whether containment in FPC implies containment in FP~\cite{RydvalDescr}.

\subsection{Finite on Top of Temporal}
In this paper, we address a class of CSPs which is of fundamental importance for advancing our understanding of the descriptive complexity of CSPs in general. 
The class naturally generalises both the class of temporal CSPs and the class of finite-domain CSPs. 

\begin{definition}\label{def:totof} 
Fix $k\geq1$ and let $B=[k]\times\Q$, where $[k]=\{1,\ldots,k\}$.
Then $C_i\coloneqq \{i\}\times\Q$ carries a strict rational order $<_i$;
the relation $\sim$ says that two elements belong to the same part.
Formally, 
\begin{alignat*}{2}
 \mathfrak D_k=(B;<_1,\ldots,<_k,\sim),\qquad & 
 (i,a)<_j(i',b)\ & &\Longleftrightarrow\ i=i'=j\text{ and }a<b \\
&  (i,a) \sim (i',b) \ & &\Longleftrightarrow\ i=i' 
\end{alignat*} 
A structure is called \emph{finite-on-top-of-temporal} if it is a first-order expansion 
of $\mathfrak D_k$ with a finite relational signature.
\end{definition} 

In the context of this definition, we call $[k]$ the \emph{colour set}, 
and if $(i,a) \in B$, then we call $i$ the \emph{colour} of $(i,a) \in C_i$. 
We may thus view $\bD_k$ as $k$ distinguished copies of $({\mathbb Q};<)$, and additionally an equivalence relation that relates elements that have the same colour. 

The class of CSPs for templates that are finite-on-top-of-temporal  
is
contained in a larger class of CSPs from the literature, and in fact is quite prototypical in this class, which will be discussed in the following. 

\subsection{Restricted Orbit Growth}
A fruitful way of obtaining classification results in model theory, 
and consequently of restricting the class of templates for CSPs, 
is to focus on structures $\B$ such that for every number of variables $n$, there is a restricted number of \emph{essentially different solutions}. 
Formally, let $f_{\B}(n)$ be the number of orbits of $\Aut(\B)$ in its natural action on the set of $n$-element subsets of $B$.

\begin{example}
    A structure $\B$ is finite if and only if $f_{\B}(n)$ is eventually $0$. 
\end{example}

\begin{example} 
A famous result of Cameron~\cite{Cameron5} states that for any countably infinite structure $\B$, we have 
$f_{\B}(n) = 1$ for all $n \in {\mathbb N}$ if and only if $\B$ is isomorphic to a temporal structure. 
\end{example} 

\begin{example} 
By the theorem of Ryll-Nardzewski, Engeler, and Svenonius (see, e.g.,~\cite{Hodges}) a countable structure $\B$ is 
$\omega$-categorical\footnote{A relational structure $\B$ with finite signature is called \emph{$\omega$-categorical} if its first-order theory, i.e., the set of all first-order sentences that hold in $\B$, has only one countable model up to isomorphism.} if and only if 
$f_{\B}(n)$ is finite for all $n \in {\mathbb N}$.
\end{example} 

It is easy to see that if $\B$ is finite-on-top-of-temporal, then $f_{\B}(n)$ is polynomial: 
if $\B$ is a first-order expansion of $\bD_k$ for some $k \geq 1$, then 
$f_{\B}(n) = f_{\bD_k}(n) =\binom{n+k-1}{k-1}$ is the number of ways how to write $n$ as a sum $n_1 + \cdots + n_k$ of non-negative integers $n_1,\dots,n_k$. 
Another example of structures $\B$ where $f_{\B}$ is polynomial are \emph{cellular structures}: this is the smallest class ${\mathcal C}$ of structures that contains countable disjoint unions of a finite structure and is closed under finite disjoint unions and finite-index first-order reducts~\cite{Bodor24}.

Countable structures where $f_{\B}(n)$ is bounded by a polynomial have been classified by Thiery and Falque~\cite{PolynomialGroups}. 
It turns out that besides the structures mentioned above 
these are precisely the structures 
that have an expansion by finitely many constants which is bidefinable with a disjoint union of a finite-on-top-of-temporal structure 
and a cellular structure~\cite[Theorem 1.6]{BODOR_2026}.

This programme sits within a broader hierarchy of orbit growth.
For countable homogeneous $\omega$-categorical relational structures,
growth bounded by $2^{O(n)}$ is conjectured to be equivalent to
\emph{monadic NIP}, meaning that NIP persists under arbitrary expansions
by unary predicates~\cite[Conjecture~1]{MonadicNIP}.
The implication from this growth bound to monadic NIP is already
known~\cite[Theorem~1.2]{MonadicNIP}.
Further up the hierarchy, every first-order reduct of a homogeneous
structure in a finite relational signature has subset-orbit growth
bounded by $2^{n^{O(1)}}$: if the maximum relation arity is $r\geq1$,
there are at most $2^{O(n^r)}$ labelled induced structures on $n$
elements, and homogeneity identifies isomorphism types with subset
orbits. The polynomial-growth case is thus a natural starting point
for a systematic classification within this much larger class.

\subsection{The Classification}
Let $\mathfrak X=(\Q;X)$, where $X(x,y,z)\ \Leftrightarrow\ x=y<z\ \lor\ y=z<x\ \lor\ z=x<y$. 
For a nontrivial abelian group $M$, let
$\mathfrak A_M=(M;R_M,(\{a\})_{a\in M})$, where
$R_M(x,y,z,t)$ means $x+y=z+t$; we call these structures \emph{affine}.
To formulate our result, we also need the concept of \emph{pp-constructions}, defined in Section~\ref{sec:preliminaries}.
Such constructions transfer CSP
definability in $\FPC$ \cite{RydvalDescr}.

\begin{theorem}\label{thm:main}
Let $\B$ be finite-on-top-of-temporal.
%a finite-signature first-order expansion of $\mathfrak D_k$. 
Then the following are equivalent.
\begin{enumerate}
    \item $\CSP(\B)$ is definable in $\FPC$; 
    \item $\B$ pp-constructs
neither $\mathfrak X$ nor any 
finite affine structure. 
    \item $\CSP(\B)$ is definable in $\FP^*$, i.e., in 
fixed-point logic with an ordered auxiliary number sort and without counting terms (see, e.g., \cite{graedel}, \cite{otto}).
\end{enumerate}
\end{theorem}

Note that in general, $\FP\subsetneq\FPnum\subsetneq\FPC$. While $\FPnum$ can express for example parity of the domain of a structure, it cannot express parity of a unary predicate interpreted in a structure, which $\FPC$ can express. 

We also give an algebraic characterization in terms of polymorphisms. Preservation
of the relation $\sim$ from Definition~\ref{def:totof} gives an induced operation
$\widetilde f \colon [k]^s\to[k]$, with
$$\widetilde f(\pi(a_1),\ldots,\pi(a_s))\coloneqq \pi(f(a_1,\ldots,a_s)),$$where $\pi((i,a))\coloneqq i$.
Since $f$ preserves the colours,  the operation $\widetilde f$ is \emph{idempotent}, i.e., satisfies $\widetilde f(c,\dots,c) = c$ for all $c \in [k]$. Let $\mathcal T$ consist of
the temporal operations $\min$, $\mi$, $\tll$, and their duals
(defined in Section~\ref{sec:preliminaries}).

\begin{theorem}\label{thm:algebraic}
Let $\B$ be finite-on-top-of-temporal. 
Then the equivalent items from Theorem~\ref{thm:main}
are also equivalent to
\begin{enumerate}[4.]
%\item $\CSP(\B)$ is definable in $\FPC$;
\item both of the following conditions hold:
\begin{enumerate}[(a)]
\item\label{cond:temporal} for each $i\in[k]$, some binary
$f_i\in\Pol(\B)$ restricts on $C_i^2$ to
$e_i\circ\mathfrak{t}_i$, where $\mathfrak{t}_i\in\mathcal T$ and
$e_i\colon C_i\hookrightarrow C_i$ is order-preserving;
\item\label{cond:colour} 
$\Pol(\B)$ contains $g_1,g_2$
such that $\widetilde g_1$ and $\widetilde g_2$ are 
3-4 weak near-unanimity operations.
\end{enumerate}
\end{enumerate}
\end{theorem}
\subsection{Algorithm and Proof Idea}
For an instance $\A$ of $\CSP(\B)$ with $n$ variables, 
let $\bS_n$ be the substructure of $\B$ with domain $[k] \times [n]$.
Note that $\bS_n$ has the property that $\A$ homomorphically maps to $\B$ if and only if it maps to $\bS_n$. 
We then use 
an algorithm which takes both $\A$ and $\bS_n$ as input and runs in polynomial time in the size of $(\A,\bS_n)$. Examples of such algorithms are the arc-consistency procedure, the path consistency procedure, AIP, BLP, AIP-BLP and various singleton versions of these algorithms. 
This technique was introduced under the name \emph{sampling} by Bodirsky, Macpherson, and Thapper~\cite{BMTsampling} 
and further developed by Bodirsky and Greiner, Mottet, and Zhuk; see Section~\ref{sect:relwork}.

In our setting, the uniform algorithm is as follows. 
Add universal constraints on all variables and on
every pair of variables. Run constraint-singleton arc-consistency
(CSAC), with an additional outer test that fixes each available 
tuple in an auxiliary pair constraint
and reruns CSAC. 
The algorithm is polynomial; we show that it has a uniform fixed-point implementation in Section~\ref{sect:FP}.

To prove the correctness of our algorithm, we use the sets of relations on $S$ computed by the algorithm. These relations need not be closed under
the polymorphisms $g_1,g_2,f_1,\dots,f_k$ witnessing the condition of Theorem~\ref{thm:algebraic}. We close the 
relations under $ g_1, g_2$ (using this closure
only in the proof). Every generated tuple lies in a finite CSAC system.
 An auxiliary finite CSP,
 invariant under the 3-4 weak near-unanimity polymorphisms, produces a \emph{single term} combining these certificates.
Its lift to the infinite template produces a finite CSAC system with one
available colour per variable. For prescribed colours, normal forms for algebraic
products \cite{Products} and temporal palette maps \cite{ZhukSingleton}
give a solution.

Our contribution is to combine the temporal
certificates across a nontrivial finite quotient while retaining their
consistency properties. Another new idea is to use algebraic products and results from~\cite{Products} to handle \emph{prescribed}
argument sorts; here the sorts themselves
are part of the solution. The classification unifies the known finite
and temporal obstructions without requiring the same temporal operation
on every part.

We present only first-order expansions of
%$(B;<_1,\ldots,<_k,\sim)$. 
$\bD_k$.
The full version will prove the corresponding
classification for arbitrary first-order reducts of this structure.

\subsection{Related Work}
\label{sect:relwork}
Bodirsky, Macpherson, and Thapper introduced sampling for infinite-domain
CSPs and combined samples with arc-consistency \cite{BMTsampling}.
Bodirsky and Greiner developed sampling for combinations of theories
\cite{BGsampling}. Mottet~\cite{Mottet} used  sampling to obtain
new algorithms for the tractable temporal CSPs described in~\cite{tcsps-journal}; Zhuk identified
singleton consistency procedures for the min, mi, and ll cases and their
duals and introduced criteria for the correctness of these algorithms using palette maps~\cite{ZhukSingleton}.

\section{Preliminaries}\label{sec:preliminaries}
\subsection{Fixed-Point Logic}
$\FP$ is first-order logic with inflationary fixed points and arbitrary
finite nesting. We use least fixed points for monotone operators.
$\FPC$ additionally has an ordered number sort $N_n=\{0,\ldots,n\}$
and counting terms on inputs of size $n$ \cite{graedel}.
$\FPnum$ is its fragment without counting terms, but retaining $N_n$ and
mixed-sort fixed points. 

\subsection{Relational Structures, CSPs and Polymorphisms}
\label{sect:pol}
A \emph{relational signature} is a set of relation symbols $R$, each with \emph{arity} $k \in {\mathbb N}_{\geq 1}$. A \emph{relational $\tau$-structure} $\mathfrak{A}$ consists of a set $A$, called the \emph{domain}, together with a relation $R^\mathfrak{A} \subseteq A^k$ for every symbol $R \in \tau$ of arity $k$. A structure is called \emph{finite} if its domain is finite. A relational $\tau'$-structure $\A'$ with the same domain as $\A$ and $\tau\subseteq\tau'$ is an \emph{expansion of $\A$} if $R^\mathfrak{A} = R^{\mathfrak{A}'}$ for all $R \in\tau$, and $\A$ is called a \emph{reduct} of $\A'$. 

An \emph{atomic formula of $\tau$} is a formula which uses only relation symbols from $\tau$ and equality. A first-order \(\tau\)-formula is built from atomic formulas of $\tau$, using Boolean connectives, i.e., negation, conjunction, and disjunction, as well as the quantifiers \(\exists x\) and \(\forall x\). The quantifiers range over elements of the domain of the structure. If $\phi(x_1,\dots,x_n)$ is a $\tau$-formula and $\A$ is a $\tau$-structure, then $\phi^{\A}$ denotes the relation $\{ (a_1,\dots, a_n) \in A^n \mid \mathfrak{A} \models \phi(a_1, \dots, a_n)\}$, called the \emph{relation defined by $\phi$ over $\A$}.
A formula is called \emph{primitive positive} if it is built exclusively from atomic formulas, existential quantifiers, and conjunctions.
A relation $R$ is \emph{first-order definable} over $\mathfrak{A}$ if there exists a first-order $\tau$-formula $\phi(x_1, \dots, x_n)$ such that $R$ equals $\phi^\A$.
If the formula $\phi$ can be chosen to be primitive positive, then we say that 
$R$ is \emph{pp-definable} over $\mathfrak{A}$.
If every new relation of an expansion $\A'$ of $\A$ is definable by a first-order formula, then $\A'$ is an \emph{fo-expansion}.
A reduct of an fo-expansion is called an \emph{fo-reduct}.
A \emph{substructure} of a $\tau$-structure $\A$ is a $\tau$-structure $\A'$ with domain $A'\subseteq A$ and $R^{\A'}=R^\A\cap A'^k$ for each $k$-ary $R\in\tau$. We also say that $\A'$ is \emph{induced on $A'$ by $\A$}, and write $\A[A']$ for $\A'$.
For relational structures $\mathfrak{A}$ and $\mathfrak{B}$ of the same signature, a homomorphism from $\mathfrak{A}$ to $\mathfrak{B}$ is a function $h\colon A \rightarrow B$ that \emph{preserves} all relations of $\mathfrak{A}$, that is, for every $R\in\tau$, if $(a_1,\dots , a_k) \in R^\mathfrak{A}$, then $(h(a_1), \dots, h(a_k)) \in R^\mathfrak{B}$. If there are homomorphisms from $\A$ to $\B$ and vice versa then we call $\A$ and $\B$ \emph{homomorphically equivalent}.
A \emph{pp-power} of $\A$ is a relational structure with domain $A^d$, for some $d\geq 1$, whose $r$-ary relations, viewed as $dr$-ary relations on $A$, are pp-definable over $\A$. We say that $\A$ \emph{pp-constructs} $\B$ if $\B$ is homomorphically equivalent to a pp-power of $\A$. Pp-constructability is transitive \cite{Wonderland}. 
%All structures we consider are $\omega$-categorical.

The \emph{constraint satisfaction problem (CSP) of $\B$} is the computational problem of deciding whether a given finite structure $\mathfrak{A}$ with the same signature as $\B$ and domain $A$, called the \emph{instance}, maps homomorphically to $\B$, which is called the \emph{template}. If a homomorphism $s\colon\mathfrak{A}\rightarrow\mathfrak{B}$ exists, then $s$ is called a \emph{solution} of $\mathfrak{A}$ for $\CSP(\mathfrak{B})$. A  \emph{constraint of $\A$} is an  atomic formula $\phi(a_1,\dots,a_s)$
whose variables $a_1,\dots,a_s$ are elements of $\A$ such that $\A \models \phi(a_1,\dots,a_s)$. 
We denote the set of all constraints of $\A$ by $\Phi$.
%\bluesout{A \emph{constraint} is an atomic formula $\phi(x_1,\ldots,x_s)=R(x_{i_1},\ldots,x_{i_r})$, where $x_1,\ldots,x_s$ are its distinct variables, listed in order of first occurrence, and $i_1,\ldots,i_r\in[s]$.}\todo{Ph: Unnötig mit der anderen Def, oder?} 
For a constraint $\phi(x_1,\dots,x_s)$ of $\A$, we put $s_\phi=s$, and $V_\phi=\{x_1,\ldots,x_s\}$.
%\blue{The elements of $A$ occurring in constraints $\phi\in\Phi$ are also called \emph{variables}.}\todo{hmh weiss nicht}
An $n$-ary \emph{polymorphism} of a relational structure $\mathfrak{A}$ is a function $f\colon A^n \rightarrow A$ such that for every $R$ in the signature of $\A$, if $\bar{a}_1,\dots,\bar{a}_n\in R^\mathfrak{A}$, then $f(\bar{a}_1,\dots,\bar{a}_n)\in R^\mathfrak{A}$. We denote by $\Pol(\mathfrak{A})$ the set of all polymorphisms of $\mathfrak{A}$, and conversely, for a set of operations $\mathscr{S}$, we write $\Inv(\mathscr{S})$ for the set of relations which are preserved by all operations in $\mathscr{S}$. The set $\Pol(\mathfrak{A})$ is a \emph{clone}, that is, a set of operations which is closed under composition and which contains all projections (see \cite{szendrei_clones}, \cite{poeschel}). A non-empty relation $R$ is pp-definable in a countable $\omega$-categorical structure $\A$ if and only if it is preserved by all polymorphisms of $\A$ (\cite{bod_nes}).
For a set of operations $\mathscr{S}$ on $A$, we say that $\mathscr{S}$ is \emph{locally closed} if for all $n\geq 1$ and all $n$-ary operations $g$ we have that $g\in\mathscr{S}$ whenever for all finite subsets $F\subseteq A^n$, there is $f\in \mathscr{S}$ such that $f$ and $g$ agree on $F$. A clone is locally closed if and only if it is the polymorphism clone of a relational structure (see, e.g.,~\cite{szendrei_clones}).

\subsection{Algebras and Terms}
A \emph{functional signature} is a set of function symbols $f$, each with an arity $r\in {\mathbb N}$. An \emph{algebra} $\mathbb A$ of  signature $\sigma$ consists of a non-empty domain $A$ and an operation $f^\mathbb A\colon A^r\to A$ for each $r$-ary symbol $f\in\sigma$.
A \emph{subalgebra} of $\mathbb A$ is an algebra $\mathbb B$ of the same signature $\sigma$ whose domain $B\subseteq A$ is non-empty and closed under all operations of $\sigma$, with $f^\mathbb B=f^\mathbb A|_{B^r}$ for every $r$-ary $f\in\sigma$. For a non-empty $X\subseteq A$, the \emph{subalgebra generated by $X$}, denoted by $\operatorname{Sg}_\mathbb A(X)$, is the smallest subalgebra of $\mathbb A$ whose domain contains $X$.

An $r$-ary operation $f$ on $A$ is \emph{idempotent} if
$f(x,\ldots,x)=x$ holds for every $x\in A$. An algebra is \emph{idempotent} if all operations from its signature are idempotent. The \emph{power} $\mathbb A^I$ of $\mathbb A$ for a finite non-empty set $I$ is an algebra with domain $A^I$ and operations defined coordinatewise as $f^{\mathbb A^I}(a_1,\dots,a_r)(i) = f^\mathbb A(a_1(i),\dots,a_r(i))$ for all $i\in I$.

A \emph{$\sigma$-term} is a variable or an expression
$f(t_1,\dots,t_r)$, where $f\in\sigma$ is $r$-ary and
$t_1,\dots,t_r$ are $\sigma$-terms. A term $t(x_1,\dots,x_n)$ is \emph{interpreted} in a $\sigma$-structure $\A$ by recursively replacing each symbol $f$ by $f^\A$. We denote the resulting operation $A^n\to A$ by $t^\A$. 

\subsection{Temporal CSPs}
A \emph{temporal structure} is a first-order reduct of $(\Q;<)$. A \emph{temporal CSP} is a problem $\CSP(\B)$ where $\B$ is a temporal structure with finite signature.
We use the following binary operations on $\Q$. The operation $\min\colon\Q^2\to\Q$ is defined by taking the minimal value of $x$ and $y$ with regard to $<$. The operation $\mi\colon\Q^2\to\Q$ is defined by
\[\mi(x,y)\coloneqq
\begin{cases}
\alpha(x) & \text{if }x=y,\\
\beta(y)  & \text{if }x>y,\\
\gamma(x) & \text{if }x<y,
\end{cases}\]where $\alpha,\beta,\gamma\colon\Q\hookrightarrow\Q$ are increasing embeddings such that $\alpha(x)<\beta(x)<\gamma(x)<\alpha(y)$ whenever $x<y$.
An \emph{ll-operation} is an operation $\tll\colon\Q^2\to\Q$ such that, for all $x,y,x',y'\in\Q$, it holds that $\tll(x,y)<\tll(x',y')$ whenever \begin{itemize}
    \item $x\leq 0$ and $x<x'$, or
    \item $x\leq 0 $ and $ x=x' $ and $ y<y'$, or
    \item $x>0 $ and $ x'>0 $ and $ y<y'$, or
    \item $x>0 $ and $ y=y' $ and $ x<x'$.
\end{itemize}
We define \emph{pp} as an arbitrary binary operation on $\Q$ such that $\operatorname{pp}(x,y)\leq\operatorname{pp}(x',y')$ if and only if either $x\leq 0$ and $ x\leq x'$, or $x>0$, $x'>0 $, and $y\leq y'$.
A binary operation $f$ on $\Q$ is called a \emph{lex-operation} if $f(x,y) < f(x',y')$ if and only if either $x<x'$, or $x=x'$ and $y<y'$.
For a $k$-ary operation $f$ on $\Q$, we call the operation given by $-f(-x_1,\dots,-x_k)$ the \emph{dual of $f$}. To distinguish them from their duals, we call $\min$, $\mi$ and $\tll$ \emph{primal}.
%\blue{\todo{Ph: Evtl streichen?}An \emph{ll-Horn clause} is a formula of the form \[ \bigvee_{i=1}^{m}x_i\neq y_i\;\lor\;\bigvee_{j=1}^{\ell}z_j<z_0\;\lor\;\left(\bigwedge_{j=1}^{\ell}z_j=z_0\right),\]where $m,\ell\geq0$, and the last disjunct may be omitted.An \emph{ll-Horn formula} is a finite conjunction of ll-Horn clauses.}
\begin{proposition}[{\cite[Propositions 25, 27]{tcsps-journal}, \cite[Lemma 12.4.4]{Book}}]\label{prop:temporal-generation}
For a temporal structure $\B$:
\begin{enumerate}
\item If $\Pol(\B)$ contains $\min$ or $\mi$, then $\Pol(\B)$ contains $\ppop$.
\item If $\Pol(\B)$ contains $\tll$, then $\Pol(\B)$ contains a lex-operation.
\item If $\Pol(\B)$ contains both $\ppop$ and a lex-operation, then $\Pol(\B)$ contains $\tll$.
\end{enumerate}
\end{proposition}

The importance of these operations in our context is that if $\B$ is a temporal structure with at least one of the polymorphisms $\min$, $\mi$, and $\tll$, then $\CSP(\B)$ can be expressed in FPC~\cite{RydvalDescr}.
\begin{theorem}[{\cite[Theorem 1.3]{RydvalDescr}}]\label{thm:temporal_descriptive}
    Let $\B$ be a temporal structure with finite relational signature. The following are equivalent: \begin{enumerate}
        \item $\CSP(\B)$ is expressible in $\FP$.
        \item $\CSP(\B)$ is expressible in $\FPC$.
        \item $\B$ does not pp-construct $(\{0,1\};\operatorname{1IN3})$ and $\mathfrak X$.
        \item $\B$ is preserved by $\min$, $\mi$, $\tll$, the dual of one of these operations, or by a constant operation.
    \end{enumerate}
\end{theorem}

\section{Nested Singleton Consistency}
Throughout the following, let $\B$ 
be finite-on-top-of-temporal, i.e., a finite-signature first-order expansion of $\D_k$,
and additionally assume that $\B$ satisfies both conditions from item (4) of Theorem \ref{thm:algebraic}.
Let $\A$ be an instance of $\CSP(\B)$, with variables $A$ and constraints $\Phi$. Our goal is to 
describe an algorithm that decides whether there exists a homomorphism from $\A$ to $\B$; equivalently, whether $\Phi$ is satisfiable in $\B$.
%prove that $\CSP(\B)$ can be solved in $\FP^*$.

\subsection{Tuber Systems}
Let $\A$ be an instance of $\CSP(\B)$. A \emph{tuber system} $L=(L_\phi)_{\phi\in\Phi}$ of $\A$ is a collection of sets such that $L_\phi\subseteq \phi^{\B}$ for $\phi\in\Phi$. Assume that every element $x\in A$ appears in unary constraints of $\Phi$; if not, add the (trivial) unary constraint $\top^\A(x)$. Define, for each $x\in A$, the set $L_x\coloneqq \bigcap_{\substack{\phi\in\Phi\\V_\phi=\{x\}}} L_\phi$, the intersection of the $L_\phi$ for all unary constraints $\phi$ on $x$. %If $x$ does not appear in any unary constraint, then set $L_x\coloneqq B$.
The tuber system $L$ is \emph{finite} if every set $L_\phi$ with $\phi\in\Phi$ is finite, and \emph{non-empty} if every set is non-empty. 
For $x_i\in A$ and $\phi\in\Phi$ with $\A\models\phi(x_1,\dots,x_i,\dots,x_s)$, we define $\pr_{x_i}(L_\phi) \coloneqq \{b\in B\mid\exists u\in L_\phi\colon u_i=b\}$. Similarly, for a tuple $t$ satisfying a constraint $\phi(x_1,\dots,x_s)$, we denote by $\pr_{x_i}(t)$ the element $t_i$. A \emph{tuber subsystem} $K\leq L$ is a tuber system of $\A$ such that for all $\phi\in\Phi$, we have $K_\phi \subseteq L_\phi$.

A tuber system $L$ of $\A$ is \emph{arc-consistent} (AC) if $\pr_x(L_\phi) = L_x$ holds for every $\phi\in\Phi$ and every $x\in V_\phi$, and \emph{constraint-singleton arc-consistent} (CSAC) if it is AC and if for each $\phi \in \Phi$ and every $u\in L_\phi$, there is an AC tuber subsystem $K^{\phi,u}\leq L$ with $K^{\phi,u}_\phi=\{u\}$. The subsystem $K^{\phi,u}$ is then called an \emph{AC-certificate} for $(\phi,u)$. We call $L$ \emph{singleton arc-consistent} if we only require AC-certificates for all values of $L_x$ for each $x\in A$. 
An instance $\A$ of $\CSP(\B)$ is \emph{(constraint-singleton) arc-consistent} if it has a non-empty (constraint-singleton) arc-consistent tuber system $L$. Arc-consistency of a tuber system $L$ can be obtained by successively deleting tuples $u$ from $L_\phi$ if there is $x\in V_\phi$ such that for some constraint $\psi$ with $x\in V_\psi$, we have $\pr_x(u)\notin\pr_x(L_\psi)$, until AC is reached. CSAC is obtained by repeatedly \emph{restricting a constraint $\phi$ from $\Phi$ to a tuple $u$} from $L_\phi$, i.e., temporarily replacing $L_\phi$ by $\{u\}$, and obtaining AC on this instance; removing $u$ from $L_\phi$ if obtaining AC yields the empty set. If obtaining (CS)AC yields the empty set, then 
there is no homomorphism from $\A$ to $\B$ compatible with the initial tuber system $L$, and the instance is rejected. The algorithm obtaining (CS)AC stops iterating when it reaches a (CS)AC system. The proof for expressibility in FP is given in Section~\ref{sect:FP}. 

A \emph{constraint walk} in $\A$ is a sequence of variables from $A$ and constraints from $\Phi$ of the form $W=x^W_0 \phi_1 x^W_1 \dots \phi_l x^W_l$, where $x^W_{j-1},x^W_j\in V_{\phi_j}$. Note that distinct variables $x_i^W, x_j^W$ of the walk might represent the same element $x\in A$ of the instance. A constraint walk $W$ is \emph{realised} in a tuber system $L$ of $\A$ if there is a sequence of values $(a_0,\dots,a_l)\in L_{x^W_0}\times\dots\times L_{x^W_l}$ such that for all $i\in[l]$, there is a tuple $t\in L_{\phi_i}$ with $\pr_{x^W_{i-1}}(t) = a_{i-1}$ and $\pr_{x^W_{i}}(t) = a_i$. A map $r\colon\{0,\dots,l\}\to  B$ with $r(i)\in L_{x_i^W}$ sending each index $i$ of the walk to the element $a_i\in L_{x^W_i}$ witnessing the above property is called a \emph{realisation of $W$ in $L$}. An arc-consistent tuber system $L$ of $\A$ is \emph{cycle-consistent} if for every $x\in A$ and every $a\in L_{x}$, every constraint walk with $x^W_0=x^W_l=x$ has a realisation $r$ in $L$ with $r(0) = r(l)=a$. Note that SAC implies cycle consistency: Since the instance obtained by restricting $x$ to $a$ has an AC subsystem $K^a\leq L$ for every $a\in L_{x}$, we can find a realisation by starting with $r(0)=a$ and taking $r({i})$ as any value $b_i\in K^a_{x^W_{i}}$ such that there is $t\in K^a_{\phi_i}$ with $\pr_{x^W_{i-1}} (t) = r({i-1})$ and $\pr_{x^W_{i}} (t) = b_i$, the existence of which is guaranteed by AC. Since we will ultimately reach $x^W_l = x$, the only choice for $r(l)$ is $a$.
\begin{theorem}[Finite cycle consistency {\cite{Kozik}}]
\label{thm:external-cycle}
Let $\A$ be a finite relational structure with 3-4 weak near-unanimity polymorphisms. Then every cycle-consistent instance of $\CSP(\A)$ has a solution.
\end{theorem}

\subsection{The Algorithm}\label{sec:algorithm}
We may assume that the instance $\A$ is not trivial, i.e., that $n\coloneqq |A| \geq 1$ and that all constraints $\phi \in \Phi$ have arity at least $1$. Denote by $S$ the finite subset of $B$ containing $n$ elements from each $C_i$, i.e., $S = \{(i,1),\dots,(i,n)\mid 1 \leq i \leq k\}$, and let $\bS\coloneqq \B[S]$ be the induced finite substructure of $\B$ on $S$. In this context, we also call the sets $(L_\phi)_{\phi\in \Phi}$ of a tuber system \emph{constraint lists}. The initial constraint lists are $L_\phi = \phi^\B \cap S^{\ar(\phi)}$, where $\ar(\phi)$ denotes the arity of $\phi$. Additionally, we add for every ordered pair of variables $x,y\in A$ a constraint $\psi_{xy}$, whose constraint list we initialize as the full binary relation $S^2$ if $x\neq y$, and as $\{(a,a)\mid a\in S\}$ if $x=y$. We denote the set of the added pair constraints by $\Psi$, and we define $L_x\coloneqq \pr_x(L_{\psi_{xx}})$. In the context of this algorithm, a tuber system $L$ of $\A$ always also contains the sets $(L_\psi)_{\psi\in\Psi}$. %\bluesout{We write $\A'$ for the instance obtained by adding the \blue{relations corresponding to the} pair constraints,
%\blue{i.e., the binary relation defined by $\top$,}}

The algorithm is described in the following.
\begin{enumerate}
\item Obtain CSAC on the current tuber system $L$. Reject if an empty set is derived.
\item For each $\psi\in\Psi$ and $u\in L_\psi$, restrict $\psi$ to $u$ in a copy of $L$ and obtain CSAC on the copy.
\begin{enumerate}
    \item If the empty set is derived, remove the tuple $u$ from $L_\psi$ and restart at step~1.
    \item If CSAC accepts the modified instance, continue.
\end{enumerate}
\item If the algorithm terminates without deriving the empty set, accept.
\end{enumerate}
If the algorithm accepts, then the accepted tuber system $L'$ is finite and CSAC. Additionally, for every $\psi\in\Psi$ and every $u\in L'_\psi$, there is a non-empty CSAC subsystem $K^{\psi,u}\leq L'$ with $K^{\psi,u}_\psi=\{u\}$, called the \emph{CSAC-certificate} for $(\psi,u)$. The algorithm terminates in polynomial time in the size of the input.

\section{Determining the Colours}\label{sec:colour}
%Let $\B$ be finite-on-top-of-temporal satisfying both conditions of item (ii) from Theorem~\ref{thm:algebraic}, and $\A$ an instance of $\CSP(\B)$ with variables $A$ and constraints $\Phi$. 
Let $\A$, $\B$ and $\Phi$ be as in the previous section. 
Let $g_1, g_2$ be the polymorphisms witnessing condition~\conditionref{cond:colour} of Theorem~\ref{thm:algebraic}, and let $\E=([k];\widetilde g_1, \widetilde g_2)$ be the finite algebra on the colour set $[k]$ and with the two operations $\widetilde g_1, \widetilde g_2$. 
\begin{lemma}\label{lem:products}
    Let $f\in\Pol(\B)$ be $s$-ary, and let $L^1,\dots,L^s$ be finite CSAC tuber systems of $\A$. Then the tuber system $H\coloneqq f(L^1,\dots,L^s)$, defined for each constraint $\phi\in\Phi$ as \[H_\phi \coloneqq f(L^1_\phi,\dots,L^s_\phi) = \left\{ f(a^1,\dots,a^s) \mid a^j \in L_\phi^j, j\in[s]\right\},\]is finite and CSAC.
\end{lemma}
\begin{proof}
    Let $L^1,\dots,L^s$ be finite CSAC tuber systems of $\A$. We define the system $P$ by \[P_\phi \coloneqq \left\{ \left((a^1_i,\dots,a^s_i)\right)_{i\in V_\phi}\mid (a_i^j)_{i\in V_\phi} \in L_\phi^j, j\in[s]\right\}\]for every $\phi\in\Phi$. Applying $f$ coordinatewise to every tuple of $P_\phi$ yields $H_\phi$. Since, for arbitrary $\phi\in\Phi$, each $a^j\in L^j_\phi$ satisfies $\phi$, and $f$ is a polymorphism of $\B$ by assumption, we get that also each $f(a^1,\dots,a^s)$ satisfies $\phi$. By construction, $P$ and $H$ are finite tuber systems.
    For every $x\in V_\phi$, we have \[ \pr_x(P_\phi) = \prod_{j=1}^s \pr_x(L^j_\phi) = \prod_{j=1}^s L_x^j = P_x,\]where $\prod$ denotes the Cartesian product of sets, and the middle equality holds since each $L^j$ is CSAC and hence AC. Similarly, \[\pr_x(H_\phi) = \left\{f(a^1_x,\dots,a_x^s) \mid a^j \in L_\phi^j, j\in[s] \right\} = f(\pr_x(P_\phi)) = f(P_x) = H_x,\]so it follows that both $P$ and $H$ are AC. 
    To prove CSAC of $H$, fix an arbitrary constraint $\phi\in\Phi$ and a tuple $a\in P_\phi$. We write $a_x$ for the tuple $(a_x^1,\dots,a_x^s)$. Note that for every $j\in[s]$, we have $a^j = (a_x^j)_{x\in V_\phi} \in L_\phi^j$. Choose for each $j$ an AC-certificate $K^j\leq L^j$ with $K_\phi^j = \{a^j\}$. As before, we define the tuber system $Q\coloneqq \prod_{j=1}^s K^j$, and $N\coloneqq f(Q)$. By the same argument as for $P$ and $H$, we get that $N$ and $Q$ are AC. Since $Q\leq P$ and $N \leq H$, and $Q_\phi = \{a\}$ and $N_\phi = \{f(a)\}$ by construction, it follows that each tuple of every constraint list of $P$ and $H$ has an AC-certificate. Thus, $P$ and $H$ are CSAC.
\end{proof}
Assume the algorithm from Section~\ref{sec:algorithm} accepts $\A$, and let $S$ be the tuber system at the end of the algorithm. For every constraint $\theta\in\Phi\cup\Psi$, let $T_\theta$ denote the smallest relation containing $S_\theta$ and closed under both $g_1$ and $g_2$. Since $g_1,g_2\in\Pol(\B)$, all tuples in $T_\theta$ satisfy $\theta$. The sets $T_x$ for all $x\in A$ are defined similarly. Furthermore, we set $D_x\coloneqq \pi(T_x)$ for all $x\in A$ and $Q_{\psi_{xy}}\coloneqq \{(\pi(a),\pi(b))\mid(a,b)\in T_{\psi_{xy}}\}$ for each ${\psi_{xy}}\in\Psi$.
\begin{lemma}\label{lem:certificates}
    \begin{enumerate}[(i)]
        \item\label{cert:colours} Each $D_x$, $x\in A$, is a non-empty subalgebra of $\E$.
        \item The tuber system $W\coloneqq \{D_x\mid x\in A\}\cup\{Q_{\psi_{xy}}\mid{\psi_{xy}}\in\Psi\}$ is SAC.
        \item\label{cert:prescribed} For every $\psi_{xy}\in\Psi$ and every $(a,b)\in Q_{\psi_{xy}}$, there is a finite CSAC tuber system $L\leq T$ with $\pi(L_x)=\{a\}$ and $\pi(L_y)=\{b\}$.
    \end{enumerate}
\end{lemma}
\begin{proof}
To prove (i), note that the projection $\pi\colon(B;g_1,g_2)\to([k];\widetilde g_1,\widetilde g_2)$ is a homomorphism, and observe that $D_x = \pi(T_x)$, where $T_x$ is the closure of the set $S_x$ under $g_1,g_2$ in $B$. Since $\widetilde g_1, \widetilde g_2$ are the only operations of $\E$ and $S_x$ is non-empty, it follows that $D_x$ is a non-empty subalgebra of $\E$.

For (ii), we first show that the tuber system $T$ is non-empty and AC. We define for every $\{g_1,g_2\}$-term $t$ with arity $q_t$ the tuber system $L^t$ via \[ L^t_\theta \coloneqq \{t^\B (u_1,\dots,u_{q_t}) \mid u_1,\dots,u_{q_t}\in S_\theta\},\] for each $\theta\in\Phi\cup\Psi$, and where $t^\B$ denotes the operation defined by $t$ in $\B$. Since $S$ is non-empty, also $L^t$ is non-empty for every $\{g_1,g_2\}$-term $t$. By Lemma~\ref{lem:products}, each $L^t$ is also finite and CSAC. For every $\theta\in\Phi\cup\Psi$, we have that $T_\theta = \cup_t L_\theta^t$. For $\theta\in\Phi\cup\Psi$, we get $\pr_x (T_\theta) = \pr_x(\cup_t L_\theta^t) = \cup_t L_x^t = T_x$, so $T$ is AC, and non-empty, since $S_\theta \subseteq T_\theta$ for each $\theta$. In particular, for every $\psi_{xy}\in\Psi$, we have $\pr_1(Q_{\psi_{xy}}) = \pi(\pr_1(T_{\psi_{xy}})) = \pi(T_x) = D_x$, and $\pr_2(Q_{\psi_{xy}})=D_y$. Thus, the tuber system $W$ is an AC tuber system. For SAC, fix $x\in A$ and a value $a\in D_x$. Choose $b\in T_x$ with $\pi(b)=a$, and choose a $\{g_1,g_2\}$-term $t$ such that $b\in L^t_x$. Since $L^t$ is CSAC, there is an AC tuber subsystem $K^t\leq L^t$ with $K_x^t=\{b\}$. For each $y\in A$, set $U_y = \pi(K^t_y)$ and $U_{\psi_{yz}} = \{(\pi(c),\pi(d))\mid (c,d) \in K^t_{\psi_{yz}}\}$ for $\psi_{yz}\in\Psi$. By construction, $U$ is non-empty and a subsystem of $W$, and since $\pr_1(U_{\psi_{yz}}) = \pi(\pr_1(K^t_{\psi_{yz}})) = \pi(K^t_{y}) = U_y$ and $\pr_2(U_{\psi_{yz}}) = U_z$, we get that $U$ is an AC-certificate, thus proving SAC of $W$.

To prove (iii), consider $(a,b)\in Q_{\psi_{xy}}$ and choose a tuple $u\in T_{\psi_{xy}}$ with $\pi(u_x)=a$ and $\pi(u_y)=b$. By construction, we obtain $u$ as $u=t^\B(u_1,\dots,u_q)$ for a $\{g_1,g_2\}$-term $t$ of arity $q$, and with $u_1,\dots,u_q\in S_{\psi_{xy}}$. Since $S$ is a tuber system the algorithm accepted, there is, for every $j\in[q]$, a finite non-empty CSAC tuber system $L^j\leq S$ with $L^j_{\psi_{xy}}=\{u_j\}$. By Lemma~\ref{lem:products}, the tuber system $L\coloneqq t^\B(L^1,\dots,L^q)$ is a finite CSAC subsystem of $T$, which by construction satisfies $L_{\psi_{xy}} = \{u\}$, and thus also $\pi(L_x)=\{a\}$ and $\pi(L_y)=\{b\}$. 
\end{proof}
%%%%%%%%%%%%%%%%%%%%%%%%%%%%%%%%%%%%%%%%%%%%%%%%%%%%%%%%%%%%%%%%%%%%%%%%%%%%%%%%%%%%%%%%%%%%%%%%%%%%%%%%%%%%%%%%%%%%%%%%%%%%%%%%
\begin{proposition}\label{prop:monochromatic}
Let $\A$ be an instance of $\CSP(\B)$ which the algorithm from Section~\ref{sec:algorithm} accepts, and let $S$ be the non-empty tuber system on which the algorithm terminates. Then there is a map $c\colon A\to[k]$ as well as a finite non-empty CSAC tuber system $U$ of $\A$ such that for every $\phi\in\Phi$, \begin{equation}\label{eq:monochromatic}
    U_\phi \subseteq \phi^\B \cap \prod_{j=1}^{s_\phi} C_{c(x_j)}
\end{equation}holds. A tuber system satisfying property~(\ref{eq:monochromatic}) is called \emph{$c$-colour-determined}.
\end{proposition}
\begin{proof}
Let $T$, $D$, $Q$ and $W$ be defined as for Lemma~\ref{lem:certificates}. For each $\psi_{xy}\in\Psi$ and each $(a,b)\in Q_{\psi_{xy}}$, we choose a CSAC tuber system $L^{xy,ab}\leq T$ such that $\pi(L^{xy,ab}_x)=\{a\}$ and $\pi(L^{xy,ab}_y)=\{b\}$. By Lemma~\ref{lem:certificates} (iii), such tuber systems exist for all constraints and values. We enumerate those systems as $L^1,\dots,L^s\leq T$.
Define $P_x \coloneqq \prod_{l=1}^s L_x^l$, and define $\Lambda\coloneqq \{(x,u)\mid x\in A, u\in P_x\}$ and $r_{(x,u)}^l=\pi(u_l)$.
% Falls ohne Algebren: R als kleinste pp defbare relation die die r tupel enthält
We define a relation $R$ on $\E^\Lambda$, a power of the finite colour algebra $\E$, by $R\coloneqq \Sg_{\E^\Lambda}(r^1,\dots,r^s)$, i.e., $R$ is the subalgebra of $\E^\Lambda$ generated by the $|\Lambda|$-tuples $r^1,\dots,r^s$. By construction, we have that $R\subseteq \prod_{(x,u)\in\Lambda}D_x$. 

We claim that for every $(x,u)\in\Lambda$, also $\pr_{(x,u)}(R)=D_x$ holds. 
Fix an arbitrary $x\in A$ and $a\in D_x$. Since \(Q_{\psi_{xx}}=\{(b,b)\mid b\in D_x\}\), we have \((a,a)\in Q_{\psi_{xx}}\). By construction of the algorithm, there exists a CSAC tuber system $L^l$ with \(\pi(L_x^l)=\{a\}\).

Take an arbitrary $u\in P_x = \prod_{j=1}^s L_x^j$. It follows that $u_l\in L_x^l$, and thus $r_{(x,u)}^l = \pi(u_l) = a$. Importantly, since $l$ does not depend on $u$, we know that also $r_{(x,u')}^l = a$ holds for all $u'\neq u$. Since $r^l\in R$ by definition, we know that $a \in \pr_{(x,u)}(R)$. Since the choice of $a\in D_x$ was arbitrary, we get $D_x \subseteq \pr_{(x,u)}(R)$, proving the claim.

Additionally, we consider the projection for pairs of distinct coordinates. For $i=(x,u)$ and $j=(y,v)$, define $\pr_{ij}(R) = \{(r_i,r_j)\mid r\in R\}$. We claim that for $x\neq y$, we have $Q_{\psi_{xy}} \subseteq \pr_{ij}(R)$, while for $x=y$ and $u\neq v$, we get $\{(a,a)\mid a\in D_x\}\subseteq \pr_{ij}(R)$. Let $x\neq y$, and consider arbitrary $(a,b)\in Q_{\psi_{xy}}$. Select a CSAC tuber system $L^l$ such that $\pi(L_x^l)=\{a\}$ and $\pi(L_y^l)=\{b\}$. Since $u\in P_x$ and $v\in P_y$, we have $u_l\in L_x^l$ and $v_l\in L_y^l$, and so the same tuple $r^l\in R$ satisfies both $r_i^l = \pi(u_l)=a$ and $r_j^l=\pi(v_l)=b$, i.e., $(a,b)\in\pr_{ij}(R)$. Since the choice of $(a,b)\in Q_{\psi_{xy}}$ was arbitrary, we have $Q_{\psi_{xy}} \subseteq \pr_{ij}(R)$. For $i=(x,u)$ and $j=(x,v)$ with $u\neq v$, choose an arbitrary $a\in D_x$, and select a CSAC tuber system $L^l$ such that $\pi(L^l_x)=\{a\}$. Since now both $u_l,v_l\in L_x^l$, we get $r_i^l = \pi(u_l) = a = \pi (v_l) = r_j^l$, proving that $(a,a)\in\pr_{ij}(R)$ for each $a\in D_x$.

We now construct an instance $\mathfrak H$ of $\CSP([k];R, (D_x)_{x\in A}, (\Delta_{D_x})_{x\in A})$, where $\Delta_{D_x}\coloneqq \{(a,a)\mid a\in D_x\}$. Note that all relations of this template are invariant under $\widetilde g_1, \widetilde g_2$, i.e., the template has 3-4 weak near-unanimity polymorphisms. The instance $\mathfrak H$ is defined as follows:\begin{itemize}
    \item The variable set of $\mathfrak H$ is $\Lambda$. Fix an enumeration $v_1,\dots,v_\lambda\in\Lambda$, where $\lambda = |\Lambda|$.
    \item For each variable  $(x,u)\in\Lambda$, add a constraint $D_x\left((x,u)\right)$.
    \item For each pair of variables $(x,u),(x,w)\in\Lambda$, add a constraint $\Delta_{D_x}\left((x,u),(x,w)\right)$.
    \item Add one $R$-constraint on all variables of $\Lambda$, i.e., $R(v_1,\dots,v_\lambda)$.
\end{itemize}
The instance $\mathfrak H$ is AC, which follows directly from our previously proved claims. 

To see that $\mathfrak H$ is also cycle consistent, note that we can use SAC of the tuber system $W$ from Lemma~\ref{lem:certificates} (ii). The only constraint of $\mathfrak H$ not having a corresponding set in $W$ is $R$, but since $Q_{\psi_{xy}}\subseteq \pr_{ij}(R)$ and all $Q_{\psi_{xy}}$ are sets in $W$, all steps possible in $W$ are possible in $R$. Fix a closed constraint walk in $\mathfrak H$ starting at $v_0=(x_0,u_0)$ and a value $a\in D_{x_0}$. Choose an AC tuber system $W'\leq W$ with $W'_{x_0}=\{a\}$ and $W'_{\psi_{xy}}\subseteq Q_{\psi_{xy}}$ for each $\psi_{xy}\in\Psi$, the existence of which is given by SAC of $W$. We set the values $d_i$ for each variable $v_i=(x_i,u_i)$ of the constraint walk recursively, beginning with $d_0\coloneqq a$. If for $v_{i+1}$ we have $x_{i+1}=x_i$, we set $d_{i+1}\coloneqq d_i$. This works for each possible constraint $\theta_{i+1}$ between $v_i$ and $v_{i+1}$: if $\theta_{i+1}$ is either of the form $D_{x_i}$ or $\Delta_{D_{x_i}}$, it follows from the definitions. If $\theta_{i+1}$ is $R$, this follows since we have shown before that $\{(a,a)\mid a\in D_x\}\subseteq \pr_{(x,u)(x,v)}(R)$. If, on the other hand, for $v_{i+1}$ we have $x_{i+1}\neq x_i$, then the constraint $\theta_{i+1}$ has to be $R$, and we set $d_{i+1}$ such that $(d_i,d_{i+1})\in W'_{\psi_{x_ix_{i+1}}}$. This choice is compatible with $R$ since $Q_{\psi_{xy}}\subseteq \pr_{ij}(R)$. At the end of the constraint walk, we get $d_m\in W'_{x_m}=W'_{x_0}=\{a\}$, proving cycle consistency of $\mathfrak H$.

By Theorem~\ref{thm:external-cycle}, the instance $\mathfrak H$ has a solution, i.e., a map $p\colon\Lambda\to[k]$ preserving all constraint relations. By construction of $\mathfrak H$, we have $p(x,u) = p(x,u')$ for all $(x,u),(x,u')\in\Lambda$. We define $c$ as this constant value for each $x$, i.e., $c(x)\coloneqq p(x,u)$. Since the relation $R$ is generated by the $|\Lambda|$-tuples $r^1,\dots,r^s$, there is an $s$-ary $\{g_1,g_2\}$-term $t$ such that, evaluated in $\E$, we have $t^\E(r^1,\dots,r^s) = (p((x,u)))_{(x,u)\in\Lambda}$. We evaluate the term $t$ in $\B$ and define $L'\coloneqq t^\B(L^1,\dots,L^s)$. Since $t^\B\in\Pol(\B)$, the tuber system $L'$ is finite and CSAC by Lemma~\ref{lem:products}. Also, for each $x\in A$ and every $b\in L'_x$, there is $u=(u_1,\dots,u_s)\in P_x$ such that $b=t^\B(u_1,\dots,u_s)$. Since \[\pi(b) = \pi(t^\B(u_1,\dots,u_s)) = t^\E(r^1_{(x,u)},\dots,r^s_{(x,u)})=p(x,u)=c(x),\]we get $L'_x\subseteq C_{c(x)}$. Since $L'\leq T$, we also get $L'_\phi\subseteq \phi^\B$ for all $\phi\in\Phi$. Defining $U$ as the restriction of $L'$ to $\Phi$ and $A$ (i.e., removing the sets $L'_\psi$ for all $\psi\in\Psi$) gives the required CSAC tuber system $U$ of $\A$, finishing the proof. 
\end{proof}

\section{Solutions from Determined Colours}\label{sec:temporal}
In this section we show that an instance with a colour-determined CSAC tuber system has a solution. The goal of this section is to prove the following lemma.

\begin{lemma}\label{lem:temporal}
Let $\B$ be finite-on-top-of-temporal satisfying the assumptions of (4) from Theorem~\ref{thm:algebraic}. Let $\A$ be an instance of $\CSP(\B)$ with constraints $\Phi$. Assume there is a map $c\colon A\to[k]$ and a finite non-empty $c$-colour-determined CSAC tuber system $S$. %\bluesout{satisfying $S_\phi\subseteq \phi^{\B}\cap\prod_{x\in V_\phi} C_{c(x)}$ for every constraint $\phi\in\Phi$}. 
Then the instance $\A$ has a solution in $\B$ with colours $c$.
\end{lemma}
The proof of this lemma will use syntactical normal forms for relations of algebraic products of $(\Q;<)$ as well as results about palette totally-symmetric maps, both are given below.
%\bluesout{We state the needed content here; Appendix~\ref{app:temporal-tools} gives the precise external statements, their hypotheses, and their specializations.}

\subsection{Syntactical Normal Forms for Algebraic Products of $(\Q;<)$}\label{sec:products}
In the following, we write $\mathfrak Q^{(d)}$ for the $d$-ary algebraic product of $(\Q;<)$, i.e., the structure $(\Q^d; <_1,\dots,<_d,=_1,\dots,=_d)$, where $u <_i v$ means $u_i < v_i$, and $u =_i v$ means $u_i = v_i$. Let $\mathfrak P$ be a first-order expansion of $\mathfrak Q^{(d)}$. Since all polymorphisms of $\mathfrak P$ preserve $=_i$ for all $i\in[d]$, the $i$-th coordinate of the output depends only on the $i$-th coordinates of the inputs. Thus, the restriction of $f\in\Pol(\mathfrak P)$ to its $i$-th coordinates, written $f|_i$, is unique and yields an operation of $(\Q;<)$.
%\bluesout{The proof and the elementary encoding of sorted relations appear in
%Appendix~\ref{sec:external-temporal}.}
\begin{theorem}[{\cite[Proposition~4.38]{Products}}]\label{thm:product-normal}
    Let $d\geq 2$, and let $\mathfrak P$ be a first-order expansion of $\mathfrak Q^{(d)}$. Let $S\subseteq[d]$ be such that \begin{itemize}
        \item for every $p\in S$ there exists $f_p\in\Pol(\mathfrak P)$ such that $f_p|_p$ is an ll-operation, and
        \item for every $p\in[d]\setminus S$ there exists $f_p\in\Pol(\mathfrak P)$ such that $f_p|_p$ is a pp-operation, but there is no $g\in\Pol(\mathfrak P)$ such that $g|_p$ is a lex-operation.
    \end{itemize}
    Then, every relation of $\mathfrak P$ can be defined by a conjunction of clauses of one of the following forms, for some $p\in[d]$:\begin{enumerate}[(a)]
        \item if $p\in[d]\setminus S$, then the clauses can be chosen to be of the form \[u_1 \neq_{i_1} v_1 \lor \dots \lor u_m \neq_{i_m} v_m \lor y_1 \neq_p x \lor \dots \lor y_k \neq_p x \lor z_1 \leq_p x \lor \dots \lor z_l \leq_p x,\]where $i_1,\dots,i_m\in S$, and
        \item if $p\in S$, then the clauses can be chosen to be clauses of the form \[x_1 \neq_{i_1} y_1 \lor \dots \lor x_m \neq_{i_m} y_m \lor z_1 <_p z_0 \lor \dots \lor z_l <_p z_0 \lor (z_0 =_p \dots =_p z_l)\]for $i_1,\dots,i_m\in S$, and where the last disjunct may be omitted.
    \end{enumerate}
\end{theorem}

\subsection{Palette Maps}
We follow the terminology from \cite{ZhukSingleton}.
\begin{definition}
	Let $k_1,\dots,k_n\in\mathbb N$, and let $m=\sum_i k_i$. A function $f\colon B^m\to B$ is \emph{$(k_1,\dots,k_n)$-block totally symmetric} if the value $f(\bar a_1,\dots, \bar a_n)$ depends only on the sets $\left(\{(\bar a_i)_j\mid j\in[k_i]\} \right)_{i\in[n]}$. 
	An $m$-tuple $t$ is a \emph{$(k_1,\dots,k_n)$-palette tuple} if every value appearing somewhere in $t$ fills some block entirely, i.e., if for all $i\in[n]$ and $j\in[k_i]$, there is $l\in[n]$ such that, for $c\coloneqq t_{k_1+\cdots+k_{i-1}+j}$, the $l$-th block of $t$ is $(c,\dots,c)$.
	A function $f\colon B^m\to B$ is \emph{$(k_1,\dots,k_n)$-palette totally symmetric} if it is $(k_1,\dots,k_n)$-block totally symmetric on all $(k_1,\dots,k_n)$-palette tuples from $B^m$.
	A function $f\colon B^m \to B$ is called a \emph{$(k_1,\dots,k_n)$-palette polymorphism} of $\B$ if, for each relation symbol $R$ and each $(k_1,\dots,k_n)$-palette tuple $(\bar a_1,\dots,\bar a_m) \in (R^\mathfrak{B})^m$, we have $f(\bar a_1,\dots,\bar a_m) \in R^\B$. We denote the set of all $(k_1,\dots,k_n)$-palette polymorphisms of $\B$ by $(k_1,\dots,k_n)$-$\Pol(\B)$.
\end{definition}
\begin{lemma}[{\cite[Corollary 6.5 and Lemma 6.12]{ZhukSingleton}}]\label{lem:mi_palette}
    Let $\B$ be a temporal structure with $\mi\in\Pol(\B)$. Then $\Pol(\B)$ contains palette totally symmetric functions of all arities and block sizes.
    %\todo{Ph: So okay, oder?}
\end{lemma}
\begin{lemma}[{\cite[Corollary~6.15 and Lemma 6.12]{ZhukSingleton}}]\label{lem:ll_palette}
    Let $\B$ be a temporal structure with $\tll\in\Pol(\B)$. Then $(k_1,\dots,k_n)$-$\Pol(\B)$ contains a $(k_1,\dots,k_n)$-palette totally symmetric function for all $k_1,\dots,k_n,n\in\mathbb N$. One possible choice for such a function is $f(\bar v_1,\dots,\bar v_n) \coloneqq e_n\left(\min(\bar v_1),\dots,\min(\bar v_n)\right)$, where $e_n\colon\Q^n\hookrightarrow\Q$ preserves the lexicographic order.
\end{lemma}
\begin{remark}\label{rem:min_palette}
	The function $\min\colon B^m\to B$ is totally symmetric, and hence also palette totally symmetric for all arities and block sizes.
\end{remark}
\begin{theorem}[{\cite[Lemma 6.13]{ZhukSingleton}}]\label{thm:palette-csac}
    Let $\B$ be a relational structure, and suppose that for all $n,k_1,\dots,k_n\in \mathbb N$ there exists a $(k_1,\dots,k_n)$-palette totally symmetric operation in $(k_1,\dots,k_n)$-$\Pol(\B)$. Then every instance $\A$ of $\CSP(\B)$ with a finite non-empty CSAC tuber system has a solution.
\end{theorem}
%%%%%%%%%%%%%%%%%%%%%%%%%%%%%%%%%%%%%%%%%%%%%%%%%%%%%%
%%%%%%%%%%%%%%%%%%%%%%%%%%%%%%%%%%%%%%%%%%%%%%%%%%%%%%
%\bluesout{The useful extra property of $L$ is \emph{equality reflection}:	if two columns of a palette matrix have equal outputs under $L$, they
%	are equal in every row. Indeed, their block minima agree, and every
%	row occurs as an entire constant block.}
\subsection{Proof of Lemma~\ref{lem:temporal}}
We now have all ingredients required to prove that any instance $\A$ of $\CSP(\B)$ with a non-empty $c$-colour-determined CSAC tuber system $S$ has a solution with colours $c\colon A\to[k]$. We refer to tuples $u\in S_\phi$ satisfying the colouring $c$ as \emph{$c$-coloured}.

\begin{proof}[Proof of Lemma~\ref{lem:temporal}]
We first construct an operation $f$ from the operations $f_1,\dots,f_k$ and show that $f$ preserves all relations on $c$-coloured palette tuples with the help of Theorem~\ref{thm:product-normal}.
We then use this map to apply Theorem~\ref{thm:palette-csac}. Assume that every $f_i|_{C_i}$ is of the form $e_i\circ \mathfrak{t}_i$
%\todo{greek letters were formulas, use small roman letters for operations} 
for some order-preserving embedding $e_i\colon C_i\hookrightarrow C_i$ and $\mathfrak{t}_i\in\{\min,\mi,\tll\}$. 
This is without loss of generality: 
If some $\mathfrak{t}_i$ is a dual of those, we reverse the order in $C_i$ and consider the primal operation. Assume also that $k\geq 2$; for $k=1$ the result is given in \cite{ZhukSingleton}.
%\todo{why needed?\\ Ph: Der zitierte Satz für die Syntax Normalform von Produkten braucht $k\geq 2$. MB: OK!}

For every constraint $\phi\in\Phi$ with variables $V_\phi=\{x_1,\dots,x_s\}$, we write $d_j \coloneqq c(x_j)$ and define a relation $\hat R_\phi\subseteq (\Q^k)^s$ by \[ \hat R_\phi(z_1,\dots,z_s) \colon\Leftrightarrow \left( (d_1,z_1[d_1]),\dots,(d_s,z_s[d_s])\in\phi^\B\right),\]where by $z[d]$ we refer to the $d$-th entry of $z\in\Q^k$. We define a structure over $\Q^k$ by $\Pf\coloneqq (\Q^k;(\hat R_\phi)_{\phi\in\Phi}, (=_i)_{i\in[k]}, (<_i)_{i\in[k]})$. Each $f_i$ induces coordinate polymorphisms for each $d\in[k]$, by taking $f_i|_d$ as $f_i|_{C_d}$, so, like in Section~\ref{sec:products}, we get coordinate-wise defined polymorphisms $w_i\in\Pol(\Pf)$ with $w_i|_i = e_i \circ \mathfrak{t}_i$. We may take each $w_i$ such that $w_i|_i = \mathfrak{t}_i$, this preserves the same relations, i.e., is still a polymorphism of $\Pf$. Define $J \coloneqq \{i\in[k] \mid \exists f\in\Pol(\Pf)\text{ such that } f|_i \text{ is a lex-operation}\}$, and write $\mathscr D_i=\{w|_i:w\in\Pol(\Pf)\}$. Since $\Pf$ is a first-order expansion of $\mathfrak Q^{(k)}$,
each $\mathscr D_i$ contains $\Aut(\Q;<)$. Additionally, each $\mathscr{D}_i$ is locally closed by \cite[Proposition~3.5]{Products}. Thus $\mathfrak T_i=(\Q;\Inv(\mathscr D_i))$ is a temporal structure with $\Pol(\mathfrak T_i)=\mathscr D_i$. Applying Proposition~\ref{prop:temporal-generation} yields an ll-operation in $\mathscr D_i$ for all $i\in J$, while for all $i\notin J$, we get $\mathfrak t_i\in\{\min,\mi\}$, and thus also a pp-operation in $\mathscr D_i$. 
By definition, these operations are
coordinates of polymorphisms of $\Pf$.
%%%%%%%%%%%%%%%%%%%%%%%%%%%%%%%%%%%%%%%%%%%%%%%%%%%
%By Proposition~\ref{prop:temporal-generation}, $J$ contains exactly those $i\in[k]$ for which there is $w\in\Pol(\Pf)$ with $w|_i = \tll$, while for all $i\notin J$, there is $w'\in\Pol(\Pf)$ with $w'|_i\in\{\min,\mi\}$. 
%%%%%%%%%%%%%%%%%%%%%%%%%%%%%%%
Theorem~\ref{thm:product-normal} now gives us for each $\hat R_\phi$ a definition of the form \[\Lambda_\phi(\bar z) \land \bigwedge_{\substack{p\notin J \\ \lambda}}\left( \chi_{\phi,p,\lambda}(\bar z) \lor \psi_{\phi,p,\lambda} (\bar z[p]) \right), \]where $\Lambda_\phi$ is preserved by $\tll$ and uses only coordinates in $J=\{p_1,\dots,p_m\}$, while each $\chi_{\phi,p,\lambda}$ is a disjunction of disequalities between same coordinate entries from $J$, and each $\psi_{\phi,p,\lambda}$ uses only values at coordinate $p$. If $J=\varnothing$, then we take $\Lambda_\phi = \top$.
Considering the tuples $\operatorname{flat}(\bar z)\in \Q^{sm}$, defined by $\operatorname{flat}(\bar z)=(z_1[p_1],\dots,z_1[p_m],\dots,z_s[p_1],\dots,z_s[p_m])$, the formula $\Lambda_\phi$ defines an ll-invariant temporal relation $T_\phi\subseteq\Q^{sm}$. Denote by $\mathfrak T$ the expansion of $(\Q;<)$ by all $T_\phi, \phi\in\Phi$. Clearly, we have $\tll\in\Pol(\mathfrak T)$. 

Fix arbitrary $t,b_1,\dots,b_t\in\mathbb N$ and write $q=\sum_i b_i$. By Lemma~\ref{lem:ll_palette}, there exists a $(b_1,\dots,b_t)$-palette totally symmetric $L\colon \Q^q\to\Q$ preserving all $T_\phi$ on $(b_1,\dots,b_t)$-palette tuples. Choose $L$ as described in Lemma~\ref{lem:ll_palette}. Let $M$ be a $q$-ary $(b_1,\dots,b_t)$-palette totally symmetric operation preserving all $\mi$-invariant temporal relations, the existence of which is guaranteed by Lemma~\ref{lem:mi_palette}. We define $G_p\colon \Q^q\to\Q$ by \[G_p\coloneqq \begin{cases}
    L\text{, if } p\in J\\ M\text{, if } p\notin J\text{ and } \mathfrak{t}_p = \mi \\ \min\text{, if } p\notin J\text{ and } \mathfrak{t}_p = \min,
\end{cases}\]
and from $G_p$ we obtain $f\colon B^q\to B$ by defining $f\left( (i,a_1),\dots,(i,a_q)\right)\coloneqq (i, G_i(a_1,\dots,a_q))$ on monochromatic tuples, and $f\left( (i_1,a_1),\dots,(i_q,a_q)\right)\coloneqq (1,0)$ for all tuples where not all $i_j$ are equal. Since all $G_p$ are $(b_1,\dots,b_t)$-palette totally symmetric, so is $f$.
%To keep both your operation on \(B\) and Theorem 19, apply the theorem to the structure \[\mathfrak B_c = \bigl(B;(R_\phi^c)_{\phi\in\Phi},(C_i)_{i\in[k]}\bigr), \qquad R_\phi^c = \phi^{\mathfrak B}\cap \prod_{x\in V_\phi}C_{c(x)}. \] This has the same domain \(B\). Each \(R_\phi^c\) has one fixed colour pattern, so every palette sequence in it has monochromatic columns. Your corrected preservation argument therefore establishes \[ f\in(b_1,\ldots,b_t)\text{-}\operatorname{Pol}(\mathfrak B_c). \] The candidate system \(S\) supplies a CSAC instance over these relations, and a solution of that instance is a solution in \(\mathfrak B\) with the prescribed colours.

In order to apply Theorem~\ref{thm:palette-csac}, we define for each $\phi\in\Phi$ a relation $\phi^\B_c \coloneqq \phi^\B \cap \prod_{x\in V_\phi} C_{c(x)}$, and a structure $\B_c\coloneqq \left(B; (\phi^\B_c)_{\phi\in\Phi}, (C_i)_{i\in[k]}\right)$. By construction, we get $f\in(b_1,\dots,b_t)$-$\Pol(\B_c)$.
We now have to verify that $f$ preserves every constraint on $(b_1,\dots,b_t)$-palette tuples. Let $\phi\in\Phi$ be arbitrary, and let $(a^1,\dots,a^q)$ be a $(b_1,\dots,b_t)$-palette tuple, with $a^r=((d_1,a_1^r),\dots,(d_s,a_s^r))\in \phi^{\mathfrak B^c}$ for each $r\in[q]$. We extend each $a^r$ to an $s$-tuple $\hat a^r$ of elements of $\Q^k$ by defining $\hat a_j^r[p]\coloneqq a_j^r$ if $p=d_j$, and $\hat a_j^r[p]\coloneqq 0$ else. By construction, we have $\hat a^r\in \hat R_\phi$ for each $r\in[q]$, and the sequence $\hat a^1,\dots,\hat a^q$ still constitutes a $(b_1,\dots,b_t)$-palette tuple. Define $h_j[p]\coloneqq G_p(\hat a_j^1[p],\dots,\hat a_j^q[p])$, and put $\bar h\coloneqq (h_1,\dots,h_s)$. We show that $\bar h \in \hat R_\phi$.

Since every $\hat a^r$ satisfies $\Lambda_\phi$, we get $\operatorname{flat}(\hat a^r)\in T_\phi$, and all $q$ of them constitute a palette tuple, so applying $L$ preserves $T_\phi$. Since $G_p=L$ for each $p\in J$, we get that $\bar h$ satisfies $\Lambda_\phi$. 

Take any clause of the form $\chi_{\phi,p,\lambda} \lor \psi_{\phi,p,\lambda}$. We know that each $\hat a^r$ satisfies this clause. Assume first that some $\hat a^r$ satisfies $\chi_{\phi,p,\lambda}$. Then there are $j,j'\in[s]$ and $p'\in J$ such that $\hat a^r_j[p'] \neq \hat a^r_{j'}[p']$. Since $p'\in J$, we have $G_{p'}=L$, and by the choice of $L$ as in Lemma~\ref{lem:ll_palette}, we get $h_j[p'] \neq h_{j'}[p']$. Since this holds for all $p'\in J$ and all $j,j'\in[s]$ with $\hat a^r_j[p'] \neq \hat a^r_{j'}[p']$, it follows that $\bar h$ satisfies $\chi_{\phi,p,\lambda}$. 
Assume now that no $\hat a^r$ satisfies $\chi_{\phi,p,\lambda}$. Define a relation $ P_{\phi,p,\lambda}\coloneqq \pr_p \left( \hat R_\phi \cap \neg \chi_{\phi,p,\lambda}\right)$. All $\hat a^r[p]$ are in $P_{\phi,p,\lambda}$. Since $\neg\chi_{\phi,p,\lambda}$ is a conjunction of equalities, the relation $P_{\phi,p,\lambda}$ is preserved by $\mathfrak t_p$, which is either $\mi$ or $\min$. It follows that $\bar h[p] = G_p(\hat a^1[p],\dots,\hat a^q[p]) \in P_{\phi,p,\lambda}$, i.e., that $\bar h$ satisfies $\psi_{\phi,p,\lambda}$.

%Since each $\psi_{\phi,p,\lambda}$ uses only values at the coordinate $p$, it follows that $\psi_{\phi,p,\lambda}$ defines an $s$-ary relation $P_{\phi,p,\lambda}$ on $(\Q;<)$. Also, since $\psi_{\phi,p,\lambda}$ uses only values at $p$ and $w_p\in\Pol(\Pf)$ satisfies $w_p|_p=\mathfrak{t}_p$, we have that $P_{\phi,p,\lambda}$ is preserved by $\mathfrak{t}_p$, which is now either $\mi$ or $\min$. 
%
Putting this together, we get $\bar h\in \hat R_\phi$, which implies, by construction, that $f(a^1,\dots,a^q)\in\phi^{\B_c}$, showing that $f$ preserves each constraint $\phi\in\Phi$ on $c$-coloured $(b_1,\dots,b_t)$-palette tuples. Since $t,b_1,\dots,b_t$ were chosen arbitrarily, we can apply Theorem~\ref{thm:palette-csac} and get a solution of the instance $\A'$ of $\CSP(\B_c)$, where $\A'$ has domain $A$ and is constructed from $\A$ by replacing each constraint $\phi$ by $\phi_c$, which corresponds to the determined colours $c$. Since each variable is assigned a value in $B$ with colour according to $c$, this homomorphism is also a solution of $\A$ for $\CSP(\B)$.
\end{proof}
\begin{proposition}[Correctness]\label{prop:algorithm}
Let $\B$ be a relational structure satisfying the conditions of (4) from Theorem~\ref{thm:algebraic}. Then the algorithm described in Section~\ref{sec:algorithm} decides $\CSP(\B)$.
\end{proposition}
\begin{proof}
Let $\A$ be an instance of $\CSP(\B)$. If a solution exists, i.e., a homomorphism $s\colon A\to B$, then, since $|A|=n$, we can use automorphisms of $\B$ to move $s(A)$ into the sample $S$ on which the algorithm operates. By construction, no values occurring in $s(A)$ will be deleted from any set of the tuber system, so the algorithm accepts.

Assume now the algorithm accepts $\A$ as an instance of $\CSP(\B)$. By Proposition~\ref{prop:monochromatic}, there is a colour-determined CSAC tuber system of $\A$, and Lemma~\ref{lem:temporal} then gives a solution of $\A$.
\end{proof}

\section{Fixed-Point Implementation}\label{sect:FP}
The algorithm described in Section~\ref{sec:algorithm} has a uniform logical implementation in $\FPnum$.
Fix a structure $\B$ which is finite-on-top-of-temporal. We will specify an $\FPnum$ sentence $\Psi$ such that a given finite 
structure $\A$ satisfies $\Psi$ if and only if $\A$ has a homomorphism to $\B$. 
A fixed first-order sentence handles
inputs with fewer than $M\geq k$ elements; this allows each colour to be represented by
one numerical variable on the remaining inputs.
The sentence $\Psi$ has three nested fixed points to perform AC deletions, singleton tests for CSAC, and colour-pair tests, respectively. 
We will first define $\Psi$ and then prove that it has the required properties. 

Let $\sigma$ be the signature of $\B$. 
Let $r_R$ be the arity of $R\in\sigma$, and put
$\mathcal J=\{R\in\sigma:r_R>0\}\mathbin{\dot\cup}\{\rho\}$,
where $\rho$ labels the auxiliary pair constraints and $r_\rho=2$.
Fix a bijection $\ell\colon\mathcal J\to[|\mathcal J|]$.
Define $M:=\max\{k,|\mathcal J|\}$.

Define the numerical first-order formulas
\[
\begin{aligned}
 \nu_0(q)&\coloneqq\neg\exists p\,(p<q),\\
 \nu_{i+1}(q)&\coloneqq\exists p\,
   \bigl(\nu_i(p)\land p<q\land\neg\exists s\,(p<s\land s<q)\bigr)
                 \qquad(0\leq i<M).
\end{aligned}
\]
Thus, $\nu_i$ defines
the number $i$ whenever it is present, and
$\operatorname{Large}\coloneqq\exists q\,\nu_M(q)$ tests whether the
input has at least $M$ elements.

Let $\mathcal F$ be the finite set of all $\sigma$-structures on
$\{1,\ldots,t\}$, for integers $0\leq t<M$, that admit a homomorphism
to $\B$. For each $\mathfrak F\in\mathcal F$, let $\theta_{\mathfrak F}$
be the first-order sentence asserting that the input is isomorphic to
$\mathfrak F$: existentially name all its elements, assert that they are
distinct and exhaust the universe, and specify the truth value of every
relation tuple.
Put
$\kappa_{\B}\coloneqq\bigvee_{\mathfrak F\in\mathcal F}\theta_{\mathfrak F}$.
This fixed first-order sentence handles all inputs with fewer than $M$
elements. 

For the remaining inputs, a \emph{value code} $v=(c,p)$ consists of
one numerical variable $c$ for its colour and one numerical variable $p$
for its \emph{rank}. 
Define $\operatorname{Value}(c,p)\coloneqq
\bigvee_{a\in[k]}
\nu_a(c) 
\land\neg\nu_0(p)$.

It is easy to see that the expansion of $\bD_k$ by the unary relations $C_1,\dots,C_k$ is homogeneous, 
and therefore has \emph{quantifier elimination}, i.e., for each $R\in\sigma$ of positive arity there exists 
a quantifier-free definition $\chi_R$ of $R^{\B}$ in $(\bD_k,C_1,\dots,C_k)$.
Translate it into a numerical formula
$\chi_R^{\#}(v_1,\ldots,v_{r_R})$ as follows. Equality becomes equality
of value codes, $\sim$ becomes equality of colour variables, $C_a$ becomes
$\operatorname{Col}_a$, and $<_a$ becomes equality of both colours to
$a$ together with the numerical comparison of ranks.

We encode the constraint lists of a tuber system, including the auxiliary
pair constraints $\psi_{xy}$ from Section~\ref{sec:algorithm}.
Throughout this
section, $L_x$ abbreviates $\pr_x(L_{\psi_{xx}})$. Thus, a non-empty
system is AC precisely when $\pr_x(L_\phi)=L_x$ for every constraint
$\phi$ containing $x$. Testing every constraint tuple also tests every
variable value, since pinning $(a,a)$ in $L_{\psi_{xx}}$ pins $x$ to $a$.
The outer tests will pin individual tuples of the pair constraints and
run CSAC on the resulting system.

For the translation of colour predicates, put
$\operatorname{Col}_a(c,p)\coloneqq\nu_a(c)$ for $a\in[k]$;
thus $\operatorname{Value}(c,p)$ is
$(\bigvee_{a\in[k]}\nu_a(c))\land\neg\nu_0(p)$.
All of these choices depend only on $\B$.

Put $m=\max_{j\in\mathcal J}r_j$. For $R\in\mathcal J\cap\sigma$,
let $I_R(x_1,\ldots,x_{r_R})$ be $R(x_1,\ldots,x_{r_R})$;
put $I_\rho=\top$ and $\chi_\rho^{\#}=\top$.
A label is represented by one numerical variable $\lambda$, using
$\operatorname{Lab}_j(\lambda)\coloneqq\nu_{\ell(j)}(\lambda)$.
A \emph{constraint code} is a tuple $\alpha=(\lambda,\mathbf x)$,
where $\mathbf x=(x_1,\ldots,x_m)$ consists of input variables.
A \emph{tuple code} is $t=(\alpha,\mathbf v)$, where
$\mathbf v=(v_1,\ldots,v_m)$ consists of value-code blocks.
The following formulas describe the constraints and their initial lists:
\[
\begin{aligned}
 \operatorname{Con}(\lambda,\mathbf x)
 &\coloneqq\bigvee_{j\in\mathcal J}\Bigl(
   \operatorname{Lab}_j(\lambda)\land I_j(x_1,\ldots,x_{r_j})
   \land\bigwedge_{r_j<i\leq m}x_i=x_1\Bigr),\\
 \operatorname{Tup}(\lambda,\mathbf x,\mathbf v)
 &\coloneqq\operatorname{Con}(\lambda,\mathbf x)
   \land\bigwedge_{i\in[m]}\operatorname{Value}(v_i)\\
 &\hspace{5mm}\land\bigwedge_{i,j\in[m]}(x_i=x_j\to v_i=v_j)\\
 &\hspace{5mm}\land\bigvee_{j\in\mathcal J}
   \bigl(\operatorname{Lab}_j(\lambda)\land
         \chi_j^{\#}(v_1,\ldots,v_{r_j})\bigr).
\end{aligned}
\]
Positions beyond the arity of a constraint repeat $x_1$ and its value.
The equality implications also ensure that repeated variables receive
equal values. In particular, the initial list of $\psi_{xx}$ consists
exactly of the diagonal pairs. Each tuple on the distinct variables of
a constraint has a unique code. The formula $\operatorname{Con}$ is
independent of $\operatorname{Tup}$, so an initially empty constraint
list still has a constraint code.

In what follows, $t,u$ denote tuples of code variables, $\alpha$ a tuple
of constraint-code variables, $x$ an input variable, and $v$ a value-code
block. Expressions such as $\alpha(t)$, $x_i(t)$ and $v_i(t)$ abbreviate
the corresponding blocks of variables. Define
\[
\begin{aligned}
 \operatorname{At}(\alpha,t)
   &\coloneqq\operatorname{Tup}(t)\land\alpha=\alpha(t),\\
 \operatorname{Inc}(\alpha,x)
   &\coloneqq\operatorname{Con}(\alpha)\land
                     \bigvee_{i\in[m]}x=x_i(\alpha),\\
 \operatorname{Val}(t,x,v)
   &\coloneqq\operatorname{Tup}(t)\land
                     \bigvee_{i\in[m]}(x=x_i(t)\land v=v_i(t)),\\
 \operatorname{Pair}(t)
   &\coloneqq\operatorname{Tup}(t)\land
                     \operatorname{Lab}_\rho(\lambda(t)).
\end{aligned}
\]
Tuple equalities are coordinatewise. Quantifiers over codes abbreviate
quantification over their input and numerical variables.

Let $D$ be a relation variable with the argument places of a tuple code;
it records the deleted tuples. Put
\[
\begin{aligned}
 \operatorname{Reject}(D)
  &\coloneqq\exists\alpha\,\bigl(\operatorname{Con}(\alpha)\land
       \forall t\,(\operatorname{At}(\alpha,t)\to D(t))\bigr),\\
 \operatorname{Bad}(D,t)
  &\coloneqq\operatorname{Tup}(t)\land
     \exists x\exists v\exists\alpha\,\bigl(
       \operatorname{Val}(t,x,v)\land\operatorname{Inc}(\alpha,x)\\
  &\hspace{24mm}\land
       \forall u\,((\operatorname{At}(\alpha,u)\land
                     \operatorname{Val}(u,x,v))\to D(u))\bigr).
\end{aligned}
\]
The first formula detects an empty remaining constraint list, including
one that was empty initially. The second detects a tuple whose value at
some variable has no support in another constraint list containing that
variable. The diagonal pair constraints supply all the information
otherwise stored in variable lists.

For tuple-code variables $t,u$, define
\[
 P_t(u)\coloneqq\operatorname{Tup}(t)\land\operatorname{Tup}(u)
                    \land\alpha(t)=\alpha(u)\land u\ne t.
\]
Adding $P_t$ to the deletion set restricts the constraint list containing
$t$ to that tuple. If $t$ has already been deleted, the resulting list
is empty.
For a formula $F$ with a deletion-predicate parameter and a formula
$\theta$ on tuple codes, $F[\theta]$ denotes substitution for that
parameter, with bound variables renamed to avoid capture. The same
convention applies to $\operatorname{Reject}(\theta)$.
For example, $Z\lor P_t$ means $w\mapsto Z(w)\lor P_t(w)$, with $w$
fresh. Pinning parameters remain free in the inner fixed points.

With $Z$ a relation variable of the same arity as $D$, define
\begin{subequations}\label{eq:fixedpoints}
\begin{equation}
 \mathcal A[D](t)\coloneqq
   [\operatorname{lfp}_{Z,u}\,
     (D(u)\lor Z(u)\lor\operatorname{Bad}(Z,u))](t).
\end{equation}
This formula performs the deletions needed for AC. Testing each tuple
of every constraint list gives
\begin{equation}
\begin{aligned}
 \mathcal C[D](t)\coloneqq
 [\operatorname{lfp}_{Z,u}\,\bigl(&D(u)\lor Z(u)\lor\mathcal A[Z](u)\\
   &{}\lor(\operatorname{Tup}(u)\land
        \operatorname{Reject}(\mathcal A[Z\lor P_u]))\bigr)](t).
\end{aligned}
\end{equation}
Finally, test each tuple of each auxiliary pair constraint by running
CSAC after pinning that tuple:
\begin{equation}
\begin{aligned}
 \mathcal O[D](t)\coloneqq
 [\operatorname{lfp}_{Z,u}\,\bigl(&D(u)\lor Z(u)\lor\mathcal C[Z](u)\\
   &{}\lor(\operatorname{Pair}(u)\land
        \operatorname{Reject}(\mathcal C[Z\lor P_u]))\bigr)](t).
\end{aligned}
\end{equation}
\end{subequations}
All deletion predicates occur positively. Since the code lengths and
the sets of labels are fixed, these are $\FPnum$ formulas. They implement
the three levels of deletion in Section~\ref{sec:algorithm}, with all
currently justified deletions performed in parallel. Define
\begin{equation}\label{eq:fp-sentence}
 \Psi\coloneqq\kappa_{\B}\lor
   \bigl(\operatorname{Large}
   %\land\eta_{\B}
   \land
         \neg\operatorname{Reject}(\mathcal O[\bot])\bigr),
\end{equation}
where $\bot$ is the always-false predicate on tuple codes. Expanding the
abbreviations yields a closed $\FPnum$ sentence in signature $\sigma$.
It uses neither counting terms nor an order on input elements.

\begin{proposition}\label{prop:fp-upper}
Under conditions~\conditionref{cond:temporal} and~\conditionref{cond:colour},
for every finite $\sigma$-structure $\A$,
$\A\models\Psi$ if and only if there is a homomorphism $\A\to\B$.
In particular, $\CSP(\B)$ is definable in $\FPnum$.
\end{proposition}
\begin{proof}
Write $n=|A|$. If $n<M$, then $\operatorname{Large}$ is false and the
definition of $\kappa_{\B}$ proves the assertion, including for the empty
input. Suppose $n\geq M$. Then $\operatorname{Large}$ is true and
$\kappa_{\B}$ is false. 

The number sort $N_n=\{0,\ldots,n\}$ contains every colour and label
code. Value codes represent the sample $S=[k]\times[n]$, and
$\chi_R^{\#}$ defines $R^{\B}\cap S^{r_R}$.
Consequently, $\operatorname{Con}$ and $\operatorname{Tup}$ encode the
constraints of $\A$, including the auxiliary pair constraints, and
their initial lists $L^0_\phi$, with nullary constraints omitted. Let $\Theta$ be this set of constraints, including
all $\psi_{xy}$ for $x,y\in A$. Identify the set of tuple codes with
$\mathcal U=\{(\phi,a):\phi\in\Theta,\ a\in L^0_\phi\}$.
For $D\subseteq\mathcal U$, write
$L^D_\phi=\{a\in L^0_\phi:(\phi,a)\notin D\}$ and
$L^D_x=\pr_x(L^D_{\psi_{xx}})$.
We identify the formulas with the relations they define on these codes.

Every non-empty AC subsystem of $L^D$ survives the deletions in
$\mathcal A[D]$, since it supplies a support for each of its tuples in
every incident constraint list. Conversely, if
$\neg\operatorname{Reject}(\mathcal A[D])$, no remaining tuple is
unsupported. The projections of all constraint lists containing $x$
therefore agree, and equal the projection of the diagonal list
$L^{\mathcal A[D]}_{\psi_{xx}}$. Thus $L^{\mathcal A[D]}$ is AC.
It is the greatest non-empty AC subsystem of $L^D$; rejection occurs
precisely when no such subsystem exists.

Similarly, every non-empty CSAC subsystem $K\leq L^D$ survives
$\mathcal C[D]$. AC deletions preserve $K$, and an AC-certificate in
$K$ for a tuple $t$ prevents rejection after adding $P_t$.
Conversely, put $Z=\mathcal C[D]$ and suppose
$\neg\operatorname{Reject}(Z)$. Closure gives $\mathcal A[Z]=Z$,
so $L^Z$ is AC. For each surviving tuple code $t$, closure also gives
$\neg\operatorname{Reject}(\mathcal A[Z\lor P_t])$.
The preceding paragraph supplies a non-empty AC subsystem of $L^Z$
whose list at $\alpha(t)$ consists only of the tuple coded by $t$.
Hence $L^Z$ is CSAC. Certificates for variable values follow by applying
this observation to diagonal pair tuples. Thus $\mathcal C[D]$ leaves
the greatest non-empty CSAC subsystem, or rejects if none exists.

Suppose first that a homomorphism $\A\to\B$ exists. By sampling, there
is one, say $h$, whose image lies in $S$. Let $K_\phi$ contain the tuple
assigned by $h$ to the variables of $\phi$, for each
$\phi\in\Theta$. This defines a non-empty CSAC subsystem of
$L^0$. By induction, none of its tuples is deleted at any stage of
$\mathcal O[\bot]$. Indeed, CSAC deletions preserve $K$.
If a pair tuple $t$ belongs to $K$, pinning $t$ also preserves $K$,
which witnesses that the ensuing CSAC test cannot reject. Every list
therefore retains its tuple from $h$, and $\A\models\Psi$.

Conversely, suppose $\A\models\Psi$, and put
$D_*=\mathcal O[\bot]$ and $L=L^{D_*}$.
Then $\neg\operatorname{Reject}(D_*)$ and $\mathcal C[D_*]=D_*$,
so $L$ is a finite non-empty CSAC tuber system. For every
$x,y\in A$ and $a\in L_{\psi_{xy}}$, let $t$ code $(\psi_{xy},a)$.
Since $t$ survives and $\operatorname{Pair}(t)$ holds, the outer
fixed-point equation implies
$\neg\operatorname{Reject}(\mathcal C[D_*\lor P_t])$.
Its remaining constraint lists form a finite non-empty CSAC subsystem
$K\leq L$ with $K_{\psi_{xy}}=\{a\}$. Thus
\begin{equation}\label{eq:fp-pair-certificates}
 \begin{gathered}
 \text{for every $x,y\in A$ and $a\in L_{\psi_{xy}}$, there is a finite}\\
 \text{non-empty CSAC subsystem $K\leq L$ with $K_{\psi_{xy}}=\{a\}$.}
 \end{gathered}
\end{equation}
With the variable lists read from the diagonal constraints, this is
exactly the certificate property used in the proof of
Proposition~\ref{prop:monochromatic}. That proof applies to $L$ and
yields a finite non-empty colour-determined CSAC tuber system.
Lemma~\ref{lem:temporal} then gives a homomorphism $\A\to\B$.
\end{proof}

%%%%%%%%%%%%%%%%%%%%%%%%%%%%%%%%%%%%%%%%%%%%%%%%%%%%%%%%%%%%%%%%%%%%%%%%%%%%%%%%%%%%%%%%%%%%%%%%%%%%%%%%%%%%%%%%%%%%%%%%%
%%%%%%%%%%%%%%%%%%%%%%%%%%%%%%%%%%%%%%%%%%%%%%%%%%%%%%%%%%%%%%%%%%%%%%%%%%%%%%%%%%%%%%%%%%%%%%%%%%%%%%%%%%%%%%%%%%%%%%%%%
\section{Obstructions and the Classification}\label{sec:converse}
In this section we show that the CSP of a structure that does not satisfy one of the conditions in (4) of Theorem~\ref{thm:algebraic} is inexpressible in $\FPC$, and then finish the proof of Theorem~\ref{thm:main} and Theorem~\ref{thm:algebraic}.
\begin{lemma}\label{lem:local_closure}
   Let $\B$ be finite-on-top-of-temporal, let $i\in[k]$, and let $t\colon C_i^2\to C_i$. Then $t\in\overline{\{f|_{C_i}\mid f\in\Pol(\B)\}}$ if and only if there exists a binary $f\in\Pol(\B) $ and an embedding $ e\colon C_i\hookrightarrow C_i$ such that $f|_{C_i}=e\circ t$. 
\end{lemma}
\begin{proof}
Suppose first that $f \in \Pol(\B)$ is such that $f|_{C_i^2}=e\circ t$, and let
$F\subseteq C_i^2$ be finite. The map sending $e(a)$ to $a$
for $a\in t(F)$ is a finite partial automorphism of $(C_i;<_i)$.
Extend it to an automorphism of this order, and act as the identity
on the other colours. The resulting map $\alpha$ is an automorphism
of $\mathfrak D_k$, hence of $\B$. Thus $\alpha\circ f$ agrees
with $t$ on $F$, proving that $t \in \overline{\{f|_{C_i}\mid f\in\Pol(\B)\}}$.

Conversely, suppose $t\in\overline{\{f|_{C_i}\mid f\in\Pol(\B)\}}$.
Take the \emph{elementary diagram} of $\B$, using constants $c_a$ for
$a\in B$, i.e., the set of all first-order sentences that hold in the expansion of $\B$ with the additional constants $c_a$ for $a \in B$. 
Then further 
expand the signature by fresh constant symbols $z_{(a,b)}$, one for every $(a,b)\in B^2$.
For every basic $r$-ary relation $R$ and all
$(a_1,\ldots,a_r),(b_1,\ldots,b_r)\in R^{\B}$, add
$R(z_{(a_1,b_1)},\ldots,z_{(a_r,b_r)})$.
Also add $z_{(a,b)}=c_{t(a,b)}$ for all $a,b\in C_i$. Let $T$ be the resulting set of first-order sentences. 
Every finite subset of $T$ is satisfiable in $\B$:
interpret the $c_a$ by $a$ and the $z_{(a,b)}$ using a global
polymorphism whose restriction agrees with $t$ at the finitely many
specified pairs. A standard model theory argument using Compactness and downward L\"owenheim--Skolem
therefore give a countable model $\mathfrak M$ of $T$ (see, e.g.,~\cite{Hodges}).
Its reduct to the signature of $\B$ is isomorphic to $\B$ by
$\omega$-categoricity of $\B$. Fix such an isomorphism $h$, and define
$f(a,b)=h(z_{(a,b)}^{\mathfrak M})$ and
$e(a)=h(c_a^{\mathfrak M})$ for $a\in C_i$.
The relation-preservation sentences make $f$ a polymorphism of $\B$.
The elementary diagram makes $e$ an increasing embedding into $C_i$,
and the added equalities give $f|_{C_i^2}=e\circ t$.
\end{proof}
%%%%%%%%%%%%%%%%%%%%%%%%%%%%%%%%%%%%%%%%%%%%%%%%%%%%%%%%%%%%%%%%%%%%%%%%%%%%%%%%%%%%%%%%%%%%%%%%%%%%%%%%%%%%%%%%%%%
\begin{proposition}[{\cite{AtseriasBulatovDawar}, \cite[Theorem 2.7, Proposition 2.8]{RydvalDescr}}]\label{prop:pp-construction_transfer}
    Let $\B$ and $\mathfrak C$ be structures with finite relational signatures such that $\B$ is pp-constructible from $\mathfrak C$. If $\CSP(\mathfrak C)$ is expressible in $\FPC$, then $\CSP(\B)$ is expressible in $\FPC$.
\end{proposition}
\begin{proposition}\label{prop:fp-inexpr}
    Let $\B$ be finite-on-top-of-temporal which does not satisfy one of the conditions
    in Theorem~\ref{thm:algebraic}. Then $\B$ pp-constructs $\mathfrak X$ or some finite affine structure, and $\CSP(\B)$ is inexpressible in $\FPC$. 
\end{proposition}
\begin{proof}
    Assume first that $\B$ does not satisfy condition~(\ref{cond:temporal}) for $i\in[k]$. By Lemma~\ref{lem:local_closure}, we get that $\mathscr C_i\coloneqq\overline{\{f|_{C_i}\mid f\in\Pol(\B)\}}$ contains no operation from $\mathcal T$. Note that $(C_i;<_i) \cong (\Q;<)$, and that every automorphism of $(C_i;<_i)$ can be extended to an automorphism of $\B$. Hence, 
    after appropriately identifying $C_i$ with ${\mathbb Q}$ we obtain
    $\Aut(\Q;<)\subseteq\mathscr C_i$, %, so every relation in $\Inv(\mathscr C_i)$ is first-order definable in $(\Q;<)$.\todo{wo wird das verwendet?}
    and since $\mathscr C_i$ is locally closed, $\Pol(\Q;\Inv(\mathscr C_i))=\mathscr C_i$. Since for each $\mathfrak t \in \mathcal{T}$ we have $\mathfrak t\notin\mathscr{C}_i$, there is some $R_\mathfrak{t}\in\Inv(\mathscr{C}_i)$ such that $R_\mathfrak t$ is not preserved by $\mathfrak t$. We take $\mathfrak G \coloneqq (\Q;<,(R_{\mathfrak t})_{\mathfrak t\in\mathcal T})$. This $\mathfrak G$ is a temporal structure with finite signature and no polymorphism from $\mathcal T$. By Theorem~\ref{thm:temporal_descriptive}, it pp-constructs $\mathfrak X$ or $(\{0,1\};\operatorname{1IN3})$, and $\CSP(\mathfrak G)$ is inexpressible in $\FPC$. 
    For every relation $R$ of $\mathfrak G$, the relation $R^{(i)} \coloneqq \{((i,a_1),\ldots,(i,a_r))\mid(a_1,\ldots,a_r)\in R\}$ is preserved by all polymorphisms of $\B$ and hence is pp-definable in $\B$. We obtain a pp-construction of $\mathfrak G$ in $\B$. Proposition~\ref{prop:pp-construction_transfer} implies that $\CSP(\B)$ is inexpressible in $\FPC$. By transitivity of pp-constructions, $\B$ pp-constructs either $\mathfrak X$ or $(\{0,1\};\operatorname{1IN3})$.
    In the former case, we are done.
    In the latter case, by \cite[Theorem~3.2.2]{Book}, it also constructs a finite affine structure.
    %Assume first that $\B$ does not satisfy condition~(\ref{cond:temporal}), i.e., there is some $i\in[k]$ such that the \bluesout{locally closed} \todo{Brauche ich hier nicht eigentich schon locally closed? MB: stimmt, dafuer brauchen wir es, und eigentlich muss man das auch zeigen! Bei Josia gilt das analoge ja nicht.} clone $\mathscr C_i$ generated by $\left\{f|_{C_i}\mid f\in\Pol(\B)\right\}$ does not contain any $\mathfrak{t}\in\mathcal T$. Define $\mathfrak G \coloneqq (\Q; \Inv(\mathscr C_i))$. Since every $f\in\Pol(\B)$ preserves $<_i$, the structure $\mathfrak G$ is a temporal structure,\todo{MB: more detail needed} but $\Pol(\mathfrak G)$ does not contain any $\mathfrak{t}\in\mathcal T$. By Theorem 1.3 from \cite{BR}, $\mathfrak G$ pp-constructs $\mathfrak X$ or $(\{0,1\};\operatorname{1IN3})$, which in turn pp-constructs all finite structures by Theorem 3.2.2 from \cite{Book}. Since a pp-construction uses only finitely many steps, there is some finite-signature reduct $\mathfrak G'$ of $\mathfrak G$ which pp-constructs $\mathfrak X$ or $(\{0,1\};\operatorname{1IN3})$. By Theorem 1.3 from \cite{BR}, $\CSP(\mathfrak G')$ is inexpressible in $\FPC$. Since $\B$ pp-constructs $\mathfrak G'$, Proposition~\ref{prop:pp-construction_transfer} gives the inexpressibility of $\CSP(\B)$ in $\FPC$, and thus also $\B$ pp-constructs $\mathfrak X$ or some finite affine structure.

    Assume now that $\B$ does not satisfy condition~(\ref{cond:colour}), i.e., there are no $g_1,g_2\in\Pol(\B)$ such that $\widetilde g_1, \widetilde g_2$ are 3-4 weak near-unanimity operations. We write $\mathscr{C}=\{\widetilde f\mid f\in\Pol(\B)\}$. Since $\sim$ is a relation of $\B$ and thus taking induced operations commutes with composition, $\mathscr C$ is a clone over $[k]$ which does not contain 3-4 weak near-unanimity operations. Write $\mathscr{H}$ for the set of ternary and quaternary operations on $[k]$ which are not in $\mathscr{C}$, i.e., $\mathscr{H}\cap\mathscr{C}=\varnothing$. For each $h\in \mathscr{H}$, choose $R_h\in\Inv(\mathscr{C})$ not preserved by $h$, and take $\mathfrak C \coloneqq ([k]; (R_h)_{h\in\mathscr{H}})$. By the choice of all $R_h$ from $\Inv(\mathscr{C})$, all ternary and quaternary polymorphisms of $\mathfrak C$ are precisely the $3$- and $4$-ary operations in $\mathscr{C}$. By Theorem 2.8 from \cite{3-4-wnu} (also see~\cite[Theorem 15.62]{GHUA}), we know that $\mathfrak C$ pp-constructs a finite affine structure and, by Theorem 10 of \cite{AtseriasBulatovDawar}, $\CSP(\mathfrak C)$ is inexpressible in $\FPC$. For every relation $R$ of $\mathfrak C$, the relation $\hat R \coloneqq \{(b_1,\dots,b_r)\mid (\pi(b_1),\dots,\pi(b_r))\in R \}$ is preserved by all polymorphisms of $\B$, and so has a pp-definition. This yields a pp-construction of $\mathfrak C$ in $\B$, so Proposition~\ref{prop:pp-construction_transfer} again gives the inexpressibility of $\CSP(\B)$ in $\FPC$, and by transitivity of pp-constructions, also $\B$ pp-constructs a finite affine structure. 
\end{proof}
\begin{proof}[Proof of Theorems~\ref{thm:main} and~\ref{thm:algebraic}]
We prove the equivalence of the four items by showing $(4)\to (3) \to (1) \to (2) \to (4)$.
By Proposition~\ref{prop:fp-upper}, if $\B$ is finite-on-top-of-temporal satisfying the conditions of (4), then $\CSP(\B)$ is definable in $\FPnum$, which is item (3). By definition, $\FPnum \subseteq \FPC$, so (3) implies (1). The implication (1) to (2) follows by contraposition and transitivity of pp-constructions, since neither $\mathfrak X$ nor any finite affine structure have a CSP expressible in $\FPC$ (\cite{RydvalDescr}, \cite{AtseriasBulatovDawar}), and the implication from (2) to (4) follow by contraposition of Proposition~\ref{prop:fp-inexpr}.
\end{proof}

%\todo[inline]{For a given $\B$, it is decidable whether $(a)$ and $(b)$ hold! Explain why...}

\begin{remark}[Decidability of the meta-problem]
Given $k$ and first-order definitions of the additional relations of $\mathfrak B$, it is decidable whether conditions (a) and (b) in Theorem~\ref{thm:algebraic} hold, and hence it is decidable whether $\operatorname{CSP}(\mathfrak B)$ is definable in FPC. This essentially follows from the results and techniques in~\cite{BPT-decidability-of-definability} and will be presented in the full version of this paper.
%To see this, first eliminate quantifiers after adding the colour predicates. The finite Ramsey \todo{Braucht das keine Reference?}theorem permits the search for witnesses to be restricted to operations whose output colours and comparisons depend only on the joint colour and order type of their inputs. For the ll operations and their duals, name the threshold\todo{"threshold point" ist undefiniert und meiner Meinung nach unklar, im Kontext des Papers} point in the relevant colour. This canonisation\todo{Auch Canonisation tauchte bisher gar nicht auf} preserves the required temporal behaviour\todo{behaviour ebenso} up to increasing output embeddings and preserves the induced operations on colours.
%\red{For each fixed arity, there are only finitely many possible behaviour\todo{What is this?} tables. Their consistency and preservation of the finitely many basic relations can be checked on finitely many input types. Consistent tables are realised by embedding the resulting countable ordered sets of output classes into the rational orders. For condition (a), test the six prescribed binary temporal behaviours on each colour, using the embedded-lift argument\todo{Was ist das für ein Argument?}. For condition (b), enumerate ternary and quaternary colour operations satisfying the required identities and test whether they admit such lifts. Thus both searches are finite and effective.}
\end{remark}

\section{Conclusion and Open Problems}\label{sec:conclusion}
We proved that every CSP for a template which is finite-on-top-of-temporal is either expressible in FP$^*$, or not even expressible in FPC. The difficult part is showing that there exists an FP$^*$ sentence
in the situation where known obstructions to FPC fail. 
Our FP$^*$ sentence implements a relatively expensive algorithm based on first sampling a sufficiently large finite substructure of the template and then running a consistency algorithm which establishes an extension of the constraint singleton arc-consistency algorithm of Zhuk. 
We do not know whether a simpler, less expensive form of consistency suffices (but our proof strategy relies on it).

The proof also leaves open whether ordinary $\FP$
suffices throughout the positive class. The full version will extend the classification
to arbitrary first-order reducts of $\mathfrak D_k$.
We mention two orthogonal directions for future research. 
\begin{itemize}
    \item Classify the CSPs for finite-on-top-of-temporal templates with respect to containment in P and NP-hardness. This is  challenging, because it is a proper generalisation of the Bulatov-Zhuk dichotomy result; in fact, already the case $k=2$ is challenging, because it requires to deal with combinations of quotient clones that cannot be solved by Datalog with algorithms for temporal CSPs that cannot be expressed in FPC; it is in general unclear which tractable templates can arise in this setting.
    \item Determine which structures of polynomial orbit growth have a CSP in FP$^*$ and in FPC.
    %\item \red{Determine which structures with a first-order interpretation in $({\mathbb Q};<)$ have a CSP in FP$^*$ and in FPC.} 
\end{itemize}

\label{maintext-end}

%\clearpage
\bibliographystyle{plainurl}
\bibliography{lit}{}

%\clearpage
%\appendix
%%%%%%%%%%%%%%%%%%%%%%%%%%%%%%%%%%%%%%%%%%%%%%%%%%%%%%%
\end{document}